\RequirePackage{tikz}
\documentclass{article}
\usepackage[hmarginratio=3:2]{geometry}
\usetikzlibrary{matrix,backgrounds}
\usepackage[utf8]{inputenc} 
\usepackage{amssymb}
\usepackage{todonotes}
\usepackage{dsfont}
\usepackage{mathrsfs}
\usepackage{mathtools}
\usepackage{enumitem}
\usepackage{booktabs}
\usepackage{algorithm}
\usepackage{algpseudocode}
\usepackage{algorithmicx}

\usepackage[]{hyperref}
\usepackage{cleveref}
\usepackage{tabularx}

\crefformat{figure}{Figure~#1}

\usepackage{amsthm}
\theoremstyle{definition}
\newtheorem{definition}{Definition}[section]
\theoremstyle{plain}
\newtheorem{theorem}{Theorem}[section]
\newtheorem{proposition}{Proposition}[section]

\newtheorem{lemma}{Lemma}[section]
\newtheorem{corollary}{Corollary}[section]
\newtheorem{rrule}{Reduction Rule}[section]
\newtheorem{claim}{Claim}[section]

\newenvironment{claimproof}{{\noindent\textit{Proof of Claim. }}}{\hfill$\blacksquare$}

\newcommand{\Oh}{\ensuremath{\mathcal{O}}}
\DeclareMathOperator{\cl}{cl}

\algnewcommand\OR{\textbf{or}\space}
\algnewcommand\ALGAND{\textbf{and}\space}
\algnewcommand\True{\texttt{True}\space}
\algnewcommand\False{\texttt{False}\space}

\algnewcommand\algorithmicforeach{\textbf{for each}}
\algdef{S}[FOR]{ForEach}[1]{\algorithmicforeach\ #1\ \algorithmicdo}
\algnewcommand{\LineComment}[1]{\State \hfill \(\triangleright\) #1}
\algrenewcommand\algorithmicprocedure{\textbf{function}}

\algtext*{EndWhile}
\algtext*{EndFor}
\algtext*{EndIf}
\algtext*{EndProcedure}

\newcommand{\PHC}{{\normalfont coNP $\subseteq$ NP/poly}} % polynomial hierarchy collapses

\newcommand{\problemdef}[3]{
  \vspace{.5ex}
  \begin{quote}
      \normalsize\textsc{#1} \smallskip \\
      \begin{tabularx}{0.9\textwidth}{@{}l@{\hspace{3pt}}X}
        \normalsize\textbf{Input:}    & \normalsize#2 \\
        \normalsize\textbf{Question:} & \normalsize#3
      \end{tabularx}
  \end{quote}
  \vspace{.5ex}
}

\begin{document}

\title{Computing Dense and Sparse Subgraphs of Weakly Closed Graphs\footnote{An extended abstract of this work appeared in the Proceedings of the~31st International Symposium on Algorithms and Computation (ISAAC~'20) held in Hong Kong, China. 
The full version contains all missing proofs, new hardness results for \textsc{2-Club} and \textsc{Dominating Clique}, and kernel lower bounds for \textsc{Independent Dominating Set}.}}

\author{Tomohiro Koana$^1$\footnote{Supported by the Deutsche Forschungsgemeinschaft (DFG), project FPTinP, NI 369/19.}
% \orcidlink{0000-0002-8684-0611}
 \and Christian Komusiewicz$^2$
 %\orcidlink{0000-0003-0829-7032} 
 \and Frank Sommer$^2$\footnote{Supported by the Deutsche Forschungsgemeinschaft (DFG), projects MAGZ, KO~3669/4-1 and EAGR, KO~3669/6-1.} %\orcidlink{0000-0003-4034-525X}
 }
\date{%
    $^1$ Algorithmics and Computational Complexity, Technische Universität Berlin, Germany\\%
    $^2$ Fachbereich Mathematik und Informatik, Philipps-Universität Marburg, Germany\\[2ex]%
}

%\title[Computing Dense and Sparse Subgraphs of Weakly Closed Graphs]{Computing Dense and Sparse Subgraphs of Weakly Closed Graphs\footnote{An extended abstract of this work appeared in the Proceedings of the~31st International Symposium on Algorithms and Computation (ISAAC~'20) held in Hong Kong, China. 
%The full version contains all missing proofs, new hardness results for \textsc{2-Club} and \textsc{Dominating Clique}, and kernel lower bounds for \textsc{Independent Dominating Set}.
%TK was supported by the Deutsche Forschungsgemeinschaft (DFG), project FPTinP, NI 369/19. 
%FS was supported by the Deutsche Forschungsgemeinschaft (DFG), projects MAGZ, KO~3669/4-1 and EAGR, KO~3669/6-1.}}

%\author[1]{\fnm{Tomohiro} \sur{Koana}\email{tomohiro.koana@tu-berlin.de}}

%\author[2]{\fnm{Christian} \sur{Komusiewicz}\email{komusiewicz@informatik.uni-marburg.de}}

%\author[2]{\fnm{Frank}~\sur{Sommer}\email{fsommer@informatik.uni-marburg.de}}

%\affil[1]{\orgdiv{Algorithmics and Computational Complexity}, \orgname{Technische Universität Berlin}, \city{Berlin}, \postcode{10623}, \state{Berlin}, \country{Germany}}

%\affil[2]{\orgdiv{Fachbereich Mathematik und Informatik}, \orgname{Philipps-Universität Marburg}, \city{Marburg}, \postcode{35037}, \state{Hessia}, \country{Germany}}

%\relatedversion{A continously updated version of the paper is available at \url{https://arxiv.org/abs/2007.05630}.}
\maketitle

\begin{abstract}
A graph~$G$ is weakly~$\gamma$-closed if every induced subgraph of~$G$ contains one vertex~$v$ such that for each non-neighbor~$u$ of~$v$ it holds that~$ \vert N(u)\cap N(v) \vert <\gamma$. The weak closure~$\gamma(G)$ of a graph, recently introduced by Fox et al.~[SIAM J.~Comp.~2020], is the smallest number such that~$G$ is weakly~$\gamma$-closed. This graph parameter is never larger than the degeneracy (plus one) and can be significantly smaller. Extending the work of Fox et al.~[SIAM J.~Comp.~2020] on clique enumeration, we show that several problems related to finding dense subgraphs, such as the enumeration of bicliques and $s$-plexes, are fixed-parameter tractable with respect to~$\gamma(G)$. Moreover, we show that the problem of determining whether a weakly $\gamma$-closed graph~$G$ has a subgraph on at least~$k$ vertices that belongs to a graph class~$\mathcal{G}$ which is closed under taking subgraphs admits a kernel with at most~$\gamma k^2$ vertices. Finally, we provide fixed-parameter algorithms for \textsc{Independent Dominating Set} and \textsc{Dominating Clique} when parameterized by~$\gamma+k$ where~$k$ is the solution size. 
Furthermore, we show that \textsc{Independent Dominating Set} does not admit a polynomial kernel for constant~$\gamma$ under standard assumptions.
\end{abstract}

%\keywords{Fixed-parameter tractability, \texorpdfstring{$c$}{c}-closure, degeneracy, clique relaxations, bicliques, dominating set}

\section{Introduction}
In the quest to design efficient algorithms for NP-hard graph problems, a very successful approach is to exploit the sparsity of input graphs: many problems that are assumed to be hard in general graphs turn out to be efficiently solvable in sparse graphs~\cite{AG09,ELS13,GKS17,K16,KS15,Luks82,See96}. One popular sparseness measure that has been used for a variety of graph problems is the \emph{degeneracy} of the input graph $G$, defined as follows. For a vertex $v \in V(G)$, let $\deg_G(v)\coloneqq  \vert N(v) \vert $ denote the degree of~$v$. 
\begin{definition}
A graph $G$ is \emph{$d$-degenerate} if one of the following two equivalent conditions holds:
\begin{itemize}
\item There exists a degeneracy ordering $\delta\coloneqq (v_1,\ldots, v_n)$ of~$G$, that is, an ordering such that $\deg_{G_i}(v_i)\le d$ where~$G_i\coloneqq G[\{v_i,\ldots, v_n\}]$  
\item Every induced subgraph $G'$ of $G$ has a vertex~$v$ with~$\deg_{G'}(v)\le d$.
\end{itemize}
The degeneracy of a graph~$G$ is the smallest integer~$d$ such that~$G$ is~$d$-degenerate.
\end{definition}
Many graph algorithms which exploit the fact that the input graph has bounded degeneracy have been proposed. For example, there is an algorithm that enumerates all maximal cliques of a graph in~$\Oh(3^{d/3}\cdot dn)$~time and performs very efficiently on real-world input instances~\cite{ELS13}. This algorithm is an FPT-algorithm for the parameter~$d$ since the exponential part of the running time depends only on~$d$. Further applications of degeneracy include FPT-algorithms for~clique relaxations~\cite{K16,KS15} and for biclique enumeration algorithms~\cite{Epp94,HM20}. Degeneracy can also be used for problems that are W[1]-hard for their standard parameterization by solution size. For example, \textsc{Dominating Set} and related problems are W[1]-hard with respect to the solution size~$k$ but have FPT-algorithms for~$d+k$~\cite{AG09,GV08,PRS12}.

In a recent work, Fox et al.~\cite{FRSWW20} proposed exploiting a different property of real-world graphs that is motivated by the triadic closure principle. This principle postulates that people in a social network which have many common friends are likely to be friends themselves. Many real-world social networks give evidence for this postulate as they contain no pair of nonadjacent vertices with many common neighbors. The degree to which a given graph adheres to the triadic closure principle can be expressed in the \emph{closure number of~$G$}, defined as follows.
\begin{definition}[\cite{FRSWW20}]
  \label{def:c}Let $\cl_G(v) \coloneqq \max_{v' \in V \setminus N[v]}  \vert N(v) \cap N(v') \vert $ denote
  the \emph{closure number} of a vertex~$v$ in a graph~$G$. A graph $G$ is
  \emph{$c$-closed} if $\cl_G(v) < c$ for all $v \in V(G)$.
  The \emph{closure number of a graph~$G$} is the smallest integer~$c$ such that~$G$ is~$c$-closed.
\end{definition}
Fox et al.~\cite{FRSWW20} showed that a~$c$-closed graph has~$\Oh(3^{c/3}\cdot n^2)$
maximal cliques. Given that all maximal cliques can be enumerated in
$\Oh(\alpha \cdot n^2)$~time, where~$\alpha$ is the number of maximal cliques~\cite{CN85}, this bound implies
that all maximal cliques of a $c$-closed graph can be enumerated
in~$\Oh^*(3^{c/3})$~time.\footnote{The~$\Oh^*$ notation hides polynomial factors in the input size.} This means that the clique enumeration problem has an FPT-algorithm with respect to the closure number of the input graph. In companion
works, we showed that several hard graph problems such as
\textsc{Independent Set}, \textsc{Dominating Set}, \textsc{Induced Matching} and \textsc{Partial Vertex Cover} admit
polynomial kernels on $c$-closed graphs when parameterized by the respective solution size~\cite{KKS20,KKNS22}.  Recently, FPT-algorithms for further problems related to \textsc{Dominating Set}  such as \textsc{Perfect Code} were obtained by Kanesh et al.~\cite{KMR+22}.
Koana and Nichterlein \cite{KN21} studied the time complexity of finding and enumerating small induced subgraphs in $c$-closed graphs.

Fox et al.~\cite{FRSWW20}  suggested a further graph parameter which combines sparseness and triadic closure, the \emph{weak closure} of a graph.
\begin{definition}[\cite{FRSWW20}]
  \label{def:gamma}
  A graph $G$ is weakly $\gamma$-closed\footnote{To avoid confusion with the closure number~$c$, we denote the weak closure by~$\gamma$ instead of~$c$.} if one of the following holds:
  \begin{itemize}
    \item
      There exists a \emph{weak closure ordering} $\sigma\coloneqq (v_1, \dots, v_n)$ of~$G$, that is, an ordering such that $\cl_{G_i}(v_i) < \gamma$ for all $i \in [n]$ where $G_i \coloneqq  G[\{ v_i, \dots, v_n \}]$.
    \item
      Every induced subgraph $G'$ of $G$ has a vertex $v \in V(G')$ such that $\cl_{G'}(v) < \gamma$.
  \end{itemize}
  The \emph{weak closure number} of a graph~$G$ is the smallest integer~$\gamma$ such that~$G$ is weakly~$\gamma$-closed.
\end{definition}

The weak closure number~$\gamma$ of a
graph~$G$ is never larger than~$d+1$ where~$d$ is the degeneracy of~$G$ and also never
larger than the closure number~$c$ of~$G$. Consequently, fixed-parameter algorithms
for~$\gamma$ are, in principle, preferable to those for the closure number~$c$ or the
degeneracy~$d$. From an application point of view, the weak closure number is also an excellent
parameter in such graphs since it tends to take on very small values in real-world social
networks~\cite{FRSWW20} (see also Table~\ref{tab:c-closure}). Fox et al.~\cite{FRSWW20}
showed that a graph has~$\Oh(3^{\gamma/3}\cdot n^2)$ many maximal cliques which, again using
known clique enumeration algorithms, gives an algorithm that enumerates all maximal cliques
in $\Oh^*(3^{\gamma/3})$~time.
Very recently, it was shown that \textsc{Dominating Set} is FPT with respect to~$\gamma+k$~\cite{LS21} and that several problems like \textsc{Connected Vertex Cover} and \textsc{Induced Matching} admit kernels of size~$k^{\Oh(\gamma)}$~\cite{KKS21}.

\paragraph{Our Results} % We apply the weak closure parameterization to
% further graph problems.
 In a
nutshell, we show that low weak closure helps in solving a variety of graph problems that are
related to searching for sparse or dense subgraphs or for sparse or dense dominating sets. Our main results for clique relaxations are listed in Table~\ref{tab:results}; our main results for variants of \textsc{Dominating Set} are listed in \Cref{tab:results2}. 

Our results improve over the state of
the art in the following sense: the best known tractability results for these problems
employ the degeneracy of the input graph as a parameter and, as discussed above, the weak
closure is essentially a smaller parameter. For some problems, we also provide results
for the $c$-closure parameter. There are two reasons for this. First, for some problems we
obtain better running time bounds for the parameter~$c$. Second, we provide some lower
bounds for the problems under consideration and, whenever possible, we provide them for
the larger closure parameter~$c$.

From a practical point of view, the most important
results are, in our opinion, the enumeration algorithms for maximal non-induced bicliques
and maximal $s$-plexes whose running times grow moderately with~$\gamma$.
Both algorithms are based on the algorithm to enumerate all maximal cliques in weakly~$\gamma$-closed graphs~\cite{FRSWW20}. Independently, Behera et al.~\cite{HR20} obtained similar results for the enumeration of maximal~$s$-plexes and further dense subgraphs parameterized by the~$c$-closure; it seems that their algorithms for $s$-plex enumeration can be adapted to parameterization by weak closure as well~\cite{HR20}. 

\begin{table}[t]
  \centering
  \caption{An overview of our results for clique relaxations. Our algorithms for \textsc{$s$-Plex} and \textsc{Non-Induced $(k_1, k_2)$-Biclique} and our $\Oh(2^{\gamma} n^{s + 3})$-time algorithm for \textsc{$s$-Defective Clique} are based on algorithms enumerating all maximal $s$-plexes (\Cref{theo-enumerate-max-s-plexes}), non-induced bicliques (\Cref{thm:nbiclique-enum}), and $s$-defective cliques (\Cref{thm:defect-clique}), respectively.}
  \begin{tabularx}{\textwidth}{lll}
    \toprule
    Problem & Result & Reference \\
    \midrule
    \textsc{Independent Set} & $\Oh(\gamma k^2)$-vertex kernel & \Cref{cor:is} \\
    \midrule
    \textsc{$s$-Plex} & W[1]-hard for $k$ even if $c = 2$ & \Cref{thm-s-plex-w-hard} \\
    & $\Oh(2^\gamma n^{2s + 1})$-time algorithm for~$s\ge 2$ & \Cref{cor:s-plex-in-fpt-time} \\ %\Cref{theo-enumerate-max-s-plexes} \\
    \textsc{$s$-Defective Clique} & W[1]-hard for $k$ even if $c = 2$ & \cite{RS08} \\
    & $\Oh(2^{\gamma} n^{s + 3})$-time algorithm & \Cref{cor-s-def-clique-gamma} \\
    & $2^{\Oh(\gamma \sqrt{s} + s \log k)} n^{\Oh(\sqrt{s})}$-time algorithm & \Cref{thm:s-def-clique-gamma} \\ %\Cref{thm:defect-clique} \\
    \textsc{2-Club} & NP-hard for $c = 4$ & \Cref{thm-2-club-nphard-constant-closure} \\
    \midrule
    \textsc{Non-Induced $(k_1, k_2)$-Biclique} & $\Oh^*(2^\gamma)$-time algorithm & \Cref{thm-non-induced-k-k-biclique-weakly-gamma} \\
    \textsc{Induced $(k, k)$-Biclique} & $\Oh^*(\gamma^{\Oh(\gamma)})$-time algorithm & \Cref{thm-induced-k-k-biclique-weakly-gamma} \\
    \textsc{Induced $(k_1, k_2)$-Biclique} & $\Oh^*(1.6107^c)$-time algorithm if~$k_1\ge 2$ & \Cref{thm-induced-k-k-biclique-c-closed} \\
    & NP-hard if~$k_1=1$ for $c = 3$ and $\gamma = 2$ & \Cref{thm-induced-k-edge-biclique-hard} \\
    & P for~$c=2$  &  \Cref{cor:induced-biclique-2-closed}\\
    & P for $k_1 = 1$ and $\gamma = 1$ & \Cref{thm-induced-biclique-dicho} \\
    & P for $k_1 \ge 2$ and $\gamma \le k_1 + 1$ & \Cref{thm-induced-biclique-dicho} \\
    & NP-hard for $k_1 \ge 2$ and $\gamma \ge k_1 + 2$ & \Cref{thm-induced-biclique-dicho} \\
    \bottomrule
  \end{tabularx}
  \label{tab:results}
\end{table}

\begin{table}[t]
  \centering
  \caption{An overview of our results for variants of \textsc{Dominating Set.}}
  \begin{tabularx}{\textwidth}{lll}
    \toprule
    Problem & Result & Reference \\
    \midrule
    \textsc{Independent Dominating Set} & $\Oh^*(((\gamma-1)/2)^k k^{2k})$-time algorithm & \Cref{theo-ids-fpt-algo} \\
    & no~$k^{\Oh(1)}$ kernel for~$\gamma=2$& \Cref{thm-ids-no-poly-gamma-2} \\
    & no~$(k+c)^{\Oh(1)}$ kernel & \Cref{thm-ids-no-poly-gamma-2} \\
    \textsc{Dominating Clique} & $\Oh^*((\gamma - 1)^{k-1})$-time algorithm & \Cref{thm-dom-clique-weakly-gamma-fpt} \\
    & NP-hard for~$c=3$ & \Cref{thm-dom-clique-nph-constant-c} \\
    & no $\Oh(k^{c-1-\epsilon})$ kernel & \Cref{thm-dom-clique-kernel-lb} \\
    \bottomrule
  \end{tabularx}
  \label{tab:results2}
\end{table}

\paragraph{Preliminaries}
%\section{Preliminaries}

For $n \in \mathds{N}$, we denote by $[n]$ the set $\{ 1, \dots, n \}$.
For a graph~$G$, we denote by $V(G)$ and $E(G)$ its vertex set and edge set, respectively.
We let~$n \coloneqq   \vert V(G) \vert $ denote the number of vertices.
Let $X \subseteq V(G)$ be a vertex set.
We let $G[X]$ denote the subgraph induced by $X$ and $G - X \coloneqq  G[V(G) \setminus X]$ the graph obtained by removing the vertices of $X$.
We denote by $N_G(X)\coloneqq \{ y \in V(G) \setminus X \mid xy \in E(G), x \in X \}$ and~$N_G[X]\coloneqq N_G(X) \cup X$, the open and closed neighborhood of $X$, respectively.
For all these notations, when $X$ is a singleton $\{ x \}$ we may write~$x$ instead of~$\{x\}$.
The maximum degree of $G$ is~$\Delta \coloneqq  \max_{v \in V(G)} \deg_G(v)$. 
The~\emph{${h}_G$-index} of a graph~$G$ is the largest integer~$h$ such that~$G$ has at least~$h$ vertices of degree at least~$h$~\cite{ES12}.
We may drop the subscript~$\cdot_G$ when it is clear from context.

Instances~$(I,k)$ of a parameterized problem consist of a classical input instance~$I$ and a parameter~$k\in \mathds{N}$. A parameterized problem is \emph{fixed-parameter tractable} if every instance~$(I, k)$ can be solved in $f(k) \cdot  \vert I \vert ^{\Oh(1)}$ time for some computable function~$f$.
An algorithm with such a running time is an \emph{FPT-algorithm}.  
A basic class of parameterized intractability is W[1]: it is widely assumed that W[1]-hard problems do not admit an FPT-algorithm. W[1]-hardness can be shown via a parameterized reduction from a W[1]-hard problem. A \emph{parameterized reduction} from a parameterized problem~$L$ to a parameterized problem~$L'$ is an algorithm that maps each instance~$(I,k)$ of~$L$ in~$f(k)\cdot \vert I \vert ^{\Oh(1)}$~time to an equivalent instance~$(I',k')$ of~$L'$ such that~$k'\le g(k)$ for some computable function~$g$.

A \emph{kernelization} is a polynomial-time algorithm which transforms every instance~$(I, k)$ into an equivalent instance $(I', k')$ such that $ \vert I' \vert  + k' \le g(k)$ for some computable function $g$. If~$g$ is a polynomial function, then we speak of a \emph{polynomial kernel}. A problem is fixed-parameter tractable if and only if it admits a kernelization. There are, however, many problems which are fixed-parameter tractable but do not admit a polynomial kernel under standard complexity-theoretic assumptions.

For more details on parameterized complexity, we refer to the standard monographs~\cite{CFK+15,DF13}.

\section{Sparse Subgraphs}
\label{sec-is-variants}

In this section we study problems that are related to finding sparse subgraphs of a given graph. The most fundamental problem in this context is the \textsc{Independent Set} problem, where one aims to find a large edgeless subgraph or, in other words, a large  set of vertices without edges between them.  

\problemdef{Independent Set}
{A graph $G$ and $k \in \mathds{N}$.}
{Is there a vertex set $S \subseteq V(G)$ such that $ \vert S \vert  \ge k$ and the vertices in $S$ are pairwise nonadjacent?}

Since~\textsc{Independent Set} is NP-hard already on graphs with maximum degree~3~\cite{GJS76}, there is no hope for FPT-algorithms for parameterization by the closure number~$c$ or the weak closure number~$\gamma$. Parameterization by~$c+k$, however, leads to fixed-parameter tractability: in a companion work~\cite{KKS20}, we provided an $\Oh(ck^2)$-vertex kernel for \textsc{Independent Set}.
% For weakly $\gamma$-closed graphs, fixed-parameter tractability of \textsc{Independent Set} with respect to $k + \gamma$ follows from the result of Bonnet et al.~\cite{BBTW19} because any weakly $\gamma$-closed graph is $P(1, \gamma, \gamma, \gamma)$-free.
% However, the running time of their algorithm is far from practical.
Here, we strengthen this result by showing that a generalization of \textsc{Independent Set} admits a polynomial kernel with respect to the parameter $k + \gamma$.

The problem that we consider is defined as follows.
Let $\mathcal{G}$ be a graph class.
We say that $\mathcal{G}$ is \emph{monotone} if~$\mathcal{G}$ is closed under vertex and edge deletions. 
That is, if~$G\in\mathcal{G}$, then for each (not necessarily induced) subgraph~$H$ of~$G$ we have~$H\in\mathcal{G}$. The aim is now to find a large induced subgraph belonging to $\mathcal{G}$.

\problemdef{$\mathcal{G}$-Subgraph}
{A graph $G$ and $k \in \mathds{N}$.}
{Is there a vertex set $S \subseteq V(G)$ with $ \vert S \vert  \ge k$ such that $G[S] \in \mathcal{G}$?}

When $\mathcal{G}$ is the class of edgeless graphs, then \textsc{$\mathcal{G}$-Subgraph} is the same as \textsc{Independent Set}.

The kernelization algorithm consists of one reduction rule that works on a weak closure ordering $(v_1, v_2, \dots, v_n)$  of the input graph~$G$.  The correctness of the reduction rule hinges on the following observation about the size of common neighborhoods of nonadjacent vertices~$v_i$ and~$v_j$ when we consider only the vertices with higher index than~$v_i$. To state this observation and the data reduction itself, let~$G_i \coloneqq G[V_i]$ for $V_i = \{ v_i, v_{i + 1}, \dots, v_n \}$.  
\begin{lemma}\label{lem:common-forward-neighbors}
    Let~$j\in [n] \setminus \{ i \}$.
    If~$v_i v_j \notin E(G)$, then~$ \vert N_{G_i}(v_i) \cap N_{G}(v_j) \vert  < \gamma$.
  \end{lemma}
  \begin{proof}
    First, assume that~$j < i$. Then, we have $ \vert N_{G_j}(v_i) \cap N_{G_j}(v_j) \vert  < \gamma$ by the definition of weak closure orderings. Since~$V_i\subseteq V_j$ this implies that~$ \vert N_{G_i}(v_i) \cap N_{G}(v_j) \vert  < \gamma$.
    Second, assume that~$j > i$. By the definition of weak closure orderings we have~$ \vert N_{G_i}(v_i) \cap N_{G_i}(v_j) \vert  < \gamma$.
  \end{proof}
  Lemma~\ref{lem:common-forward-neighbors} allows us to show the correctness of the following reduction rule, which removes vertices with many neighbors that have higher index in the weak closure ordering.
  
\begin{rrule} \label{rule:removevi}
  If $\deg_{G_i}(v_i) \ge \gamma k$, then remove $v_i$.
\end{rrule}
\begin{lemma}
  \Cref{rule:removevi} is correct for monotone graph classes.
\end{lemma}
\begin{proof}
  Let $G'\coloneqq  G - v_i$ for $v_i \in V$ with $\deg_{G_i}(v_i) \ge \gamma k$ be the graph obtained by applying Reduction Rule~\ref{rule:removevi}.
  Clearly, if~$G'[S]\in\mathcal{G}$ for some vertex set~$S\subseteq V(G')$, then also~$G[S]\in\mathcal{G}$. 
  
  Hence, it remains to show that if there is a vertex set~$S\subseteq V(G)$ of size $k$ such that~$G[S]\in\mathcal{G}$, then there is a vertex set~$S'\subseteq V(G')$ of size $k$ such that~$G'[S']\in\mathcal{G}$. 
  If~$v_i \notin S$, we observe that~$G'[S] \in \mathcal{G}$.
  Thus, in the following we assume that~$v_i \in S$.  Let~$S_i \coloneqq  S \setminus N_G(v_i)$ be the set of vertices in $S$ that are not adjacent to $v_i$. We show that there is some vertex~$u\notin S$ that is not adjacent to any vertex of~$S_i$.  By \Cref{lem:common-forward-neighbors}, any vertex~$v_j \in S_i$  has less than~$\gamma$ neighbors in~$N_{G_i}(v_i)$.
  Since~$\deg_{G_i}(v_i) \ge \gamma k$ and~$\vert S_i\vert<k$, there exists at least one vertex~$u$ in~$N_{G_i}(v_i)$ that is not adjacent to any vertex from $S_i$. Consequently,~$N_G(u)\cap S\subseteq N_G(v_i)\cap S$.
	Since $\mathcal{G}$ is monotone, we may thus replace~$v_i$ in~$S$ with~$u$: for~$S' \coloneqq (S \setminus \{ v_i \}) \cup \{ u \}$ we have $G'[S'] \in \mathcal{G}$.
 % For the other direction, observe that if $G'[S] \in \mathcal{G}$ holds for a vertex set $S \subseteq V(G')$, then $G[S] \in \mathcal{G}$ also holds.
\end{proof}

\begin{theorem}
  \label{theorem:monotonesubgraph}
  Let $\mathcal{G}$ be a monotone graph class.
  Then, \textsc{$\mathcal{G}$-Subgraph} has a kernel with at most $\gamma k^2$ vertices.
\end{theorem}
\begin{proof}
  One can exhaustively apply \Cref{rule:removevi} in polynomial time.
  The resulting graph has a weak closure ordering where every vertex~$v_i$ has less than~$\gamma k$ neighbors in~$G_i$. Hence, this graph has degeneracy~$d < \gamma k$.
  Note that any graph~$G$ on at least $(d + 1) k$~vertices contains an independent set $S$ of size $k$.
  Due to the monotonicity of~$\mathcal{G}$, $G[S] \in \mathcal{G}$ for an independent set $S$.
  Thus, returning Yes is correct whenever $ \vert V(G) \vert  \ge \gamma k^2$ and we obtain an equivalent instance with at most $\gamma k^2$ vertices.
\end{proof}

Since the class of edgeless graphs is monotone, we obtain the following.

\begin{corollary}
  \label{cor:is}
  \textsc{Independent Set} has a kernel with at most~$\gamma k^2$ vertices.
\end{corollary}

\Cref{theorem:monotonesubgraph} also implies kernels for many other problems, including \textsc{Acyclic Subgraph}, \textsc{Bipartite Subgraph}, \textsc{Planar Subgraph}, and \textsc{Bounded Degree Subgraph}.
These problems ask whether the input graph~$G$ contains a vertex set $S \subseteq V(G)$ such that $ \vert S \vert  \ge k$ and $G[S]$ is acyclic, bipartite, planar, or has bounded maximum degree, respectively.
All of these problems are W[1]-hard in general graphs~\cite{KR02}.

\begin{corollary}
Each of  \textsc{Acyclic Subgraph}, \textsc{Bipartite Subgraph}, \textsc{Bounded Degree Subgraph}, and \textsc{Planar Subgraph} has a kernel with $\gamma k^2$ vertices.
\end{corollary}

Moreover, it follows from \Cref{theorem:monotonesubgraph} that \textsc{Sparsest-$k$-Subgraph}, the problem of finding an induced subgraph with exactly~$k$ vertices and at most $t$~edges, also admits a polynomial kernel in weakly $\gamma$-closed graphs.

\begin{corollary}
  \textsc{Sparsest-$k$-Subgraph} has a kernel with at most $\gamma k^2$ vertices.
\end{corollary}
This is in sharp contrast to \textsc{Densest-$k$-Subgraph}, where one asks for a set~$S\subseteq V(G)$ of exactly~$k$ vertices such that~$G[S]$ has at least $t$ edges:  \textsc{Densest-$k$-Subgraph} is W[1]-hard with respect to~$k$ even in $2$-closed graphs~\cite{RS08}. 

%\problemdef
%{Dense Subgraph}
%{A graph $G$ and $k, t \in \mathds{N}$.}
%{Is there a set $S \subseteq V(G)$ of at least $k$ vertices such that $G[S]$ has at least $t$ edges?}

\section{Clique Relaxations}

In this section, we present algorithms for generalizations of the \textsc{Clique} problem. 
In contrast to the variants of \textsc{Independent Set} considered in \Cref{sec-is-variants}, here we only consider parameterization by the weak closure number~$\gamma$.
Recall, that Fox et al.~\cite{FRSWW20}
showed that a graph has~$\Oh(3^{\gamma/3}\cdot n^2)$ many maximal cliques.
Using known clique enumeration algorithms this gives an algorithm that enumerates all maximal cliques
in $\Oh^*(3^{\gamma/3})$~time.

\subsection{\texorpdfstring{$s$}{s}-Plex}

A clique is a vertex set~$S$ such that each vertex~$v\in S$ is adjacent to each other vertex in~$S$. 
One way to relax the clique definition is to allow each vertex~$v\in S$ to have at most~$s$ non-neighbors in~$S$. 
This idea can be formalized as follows.

\begin{definition}
In a graph~$G=(V,E)$ a set~$S\subseteq V$ is an \emph{$s$-plex} if every vertex in~$G[S]$ has degree at
least~$\vert S\vert - s$ in~$G[S]$.
\end{definition}

Observe that cliques are exactly the 1-plexes. 
Here, we study the task of enumerating maximal $s$-plexes which has received some interest in practice~\cite{CFMPT17,CMSGMV18}, and the problem of finding a sufficiently large $s$-plex, defined as follows.
\problemdef{$s$-Plex}
{A graph $G$ and $k \in \mathds{N}$.}
{Does~$G$ contain an~$s$-plex~$S$ of size at least $k$?}

On the negative side,~\textsc{$s$-Plex} is W[1]-hard when parameterized by~$k$ for all~$s \in \mathds{N}$~\cite{KR02,KHMN09}. On the positive side, a simple algorithm can enumerate all maximal~$s$-plexes of a $d$-degenerate graph in $2^d n^{s + \Oh(1)}$~time~\cite{K16}. 

For the problem of enumerating all maximal $s$-plexes, we obtain an FPT-algorithm for the weak closure number. 
\begin{theorem}
  \label{theo-enumerate-max-s-plexes}
  For $s \ge 2$,  a graph~$G$ has $\Oh(2^{\gamma} n^{2s - 1})$ maximal $s$-plexes.
  Moreover, all maximal $s$-plexes of~$G$ can be enumerated in $\Oh(2^\gamma n^{2s + 1})$ time.
\end{theorem}
\begin{proof}
First, we show the bound on the number of maximal $s$-plexes in a weakly~$\gamma$-closed graph.
  Let $v \in V(G)$ be a vertex such that $\cl_G(v) < \gamma$ and let $G' \coloneqq  G - v$ be the graph obtained by deleting $v$.
  Let $\mathcal{S}$ and $\mathcal{S}'$ be the collections of all maximal $s$-plexes (without duplicates) in $G$ and $G'$, respectively.
  We show that $ \vert \mathcal{S} \vert  \le  \vert \mathcal{S}' \vert  + 2^\gamma n^{2s-2}$ and that $\mathcal{S}$ can be constructed from $\mathcal{S}'$ in $\Oh( \vert \mathcal{S}' \vert  \cdot n + 2^\gamma n^{2s + 1})$~time.
  To obtain the bound we identify the following four types of maximal $s$-plexes in~$G$:
  \begin{enumerate}[label=Type \arabic*:,leftmargin=*]
    \item
      $S$ does not contain $v$. Then, $S$ is also maximal in $G'$.
    \item
      $S$ contains $v$ and $S \setminus \{ v \}$ is maximal in $G'$.
    \item
      $S$ contains $v$,~$S \setminus \{ v \}$ is not maximal in $G'$, and $S$ contains a non-neighbor of $v$ (that is, $S \setminus N_G(v) \ne \emptyset$).
    \item
      $S$ contains $v$,~$S \setminus \{ v \}$ is not maximal in $G'$, and $S$ is contained in the neighborhood of~$v$, that is,~$S \subseteq N_G[v]$.
  \end{enumerate}

  Clearly, each maximal $s$-plex is of one of these four types.
  It is easy to see that there are $ \vert \mathcal{S}' \vert $ maximal $s$-plexes of Type 1 and Type 2.
  Hence, it remains to bound the number of maximal $s$-plexes of Type ~3 and Type~$4$. 
  
  Next, we bound the number of maximal $s$-plexes of Type~$3$. 
  Consider such an $s$-plex~$S$.
  We may partition~$S$ into three parts as follows:
  We first divide~$S$ into~$S_v \coloneqq  S \cap N_G[v]$ and~$\widetilde{S_{v}} \coloneqq  S \setminus N_G[v]$.
  We divide~$S_v$ further into~$S_{uv} \coloneqq  S_v \cap N_G(u)$ and~$\widetilde{S_{uv}} \coloneqq  S_v \setminus N_G(u)$ for some vertex~$u \in \widetilde{S_v}$.
  Here,~$u$ is any non-neighbor of~$v$ to exploit the weak~$\gamma$-closure.
  By the definition of $s$-plexes,~$ \vert \widetilde{S_v} \vert  < s$ and~$ \vert \widetilde{S_{uv}} \vert  < s$.
  Hence, there are at most~$n^{2s - 2}$ choices for~$\widetilde{S_v}$ and~$\widetilde{S_{uv}}$.
  For~$S_{uv}$, there are at most~$2^{\gamma - 1}$ choices because~$S_{uv} \subseteq N_G(v) \cap N_G(u)$ and~$ \vert N_G(v) \cap N_G(u) \vert  < \cl_G(v) < \gamma$.
  Overall, there are at most~$2^{\gamma - 1} n^{2s - 2}$ maximal $s$-plexes of Type 3.

  It remains to bound the number of maximal $s$-plexes of Type 4.
  Let~$S$ be one of these $s$-plexes.
  Since~$S' \coloneqq  S \setminus \{ v \}$ is not maximal in~$G'$, there exists a vertex~$u \in V(G) \setminus S$ such that~$S' \cup \{ u \}$ is an $s$-plex in~$G'$.
  If~$u \in N_G(v)$, then~$S \cup \{ u \}$ is also an~$s$-plex in~$G$, which contradicts the fact that~$S$ is maximal in~$G$.
  Hence, we can assume that~$u \notin N_G(v)$.
  Then,~$S \setminus N_G(u)$ contains at most~$s - 1$ vertices, which in turn implies that there are at most~$n^{s - 1}$ choices for~$S \setminus N_G(u)$.
  Since~$S\subseteq N(v)$ we observe that~$S \cap N_G(u) \subseteq N_G(v) \cap N_G(u)$ and~$ \vert N_G(v) \cap N_G(u) \vert  \le \cl_G(v) < \gamma$. 
  Thus, we have~$2^{\gamma - 1}$ choices for~$S \cap N_G(v)$.
  All in all, there are at most~$2^{\gamma - 1} n^s$ maximal $s$-plexes of Type 4.
  
By the above analysis, we obtain~$ \vert \mathcal{S} \vert  \le  \vert \mathcal{S}' \vert  + 2^{\gamma - 1} n^{2s - 2} + 2^{\gamma - 1} n^s \le  \vert \mathcal{S}' \vert  + 2^\gamma n^{2s - 2}$.
Next, we bound the overall number of maximal $s$-plexes in a graph with~$n$ vertices.
To this end, let~$a_n$ be the number of maximal $s$-plexes in weakly $\gamma$-closed graphs on $n$ vertices.
  Clearly,~$a_1 = 1$. Furthermore, the above analysis showed that~$a_{n} - a_{n - 1} =  \vert \mathcal{S} \vert  -  \vert \mathcal{S}' \vert  \le 2^{\gamma} n^{2s - 2}$.
  Hence, by induction we obtain~$a_n = a_1 + \sum_{i = 2}^{n} (a_{i} - a_{i - 1}) \le 2^\gamma n^{2s - 1} + 1$.
In other words, a weakly~$\gamma$-closed graph on~$n$ vertices has at most~$2^\gamma n^{2s - 1} + 1$ maximal $s$-plexes.

Second, we bound the overall time needed to enumerate all maximal $s$-plexes.
To obtain this bound, we again let~$v\in V(G)$ be any vertex such that~$\cl_G(v) < \gamma$, let~$G' \coloneqq  G - v$ be the graph obtained by deleting $v$, and let~$\mathcal{S}$ and~$\mathcal{S}'$ be the collections of all (without duplicates) maximal~$s$-plexes in~$G$ and~$G'$.
Observe that all maximal~$s$-plexes of Type 1 and 2 can be found in~$\Oh( \vert \mathcal{S}' \vert \cdot n)$~time.  
Furthermore, maximal~$s$-plexes of Type 3 and 4 can be enumerated in $\Oh((2^{\gamma-1}n^{2s-2}+2^{\gamma-1}n^s) \cdot n^2)$~time, because it takes $\Oh(n^2)$ time to verify whether a vertex set is a maximal $s$-plex or not.
  Finally, we remove duplicates in $\Oh(( \vert \mathcal{S}' \vert +2^{\gamma-1}n^{2s-2}+2^{\gamma-1}n^s)\cdot n)=\Oh( \vert \mathcal{S}' \vert \cdot n+2^\gamma n^{2s-1})$~time, using radix sort.
  Altogether, the algorithm needs $\Oh( \vert \mathcal{S}' \vert \cdot n + 2^\gamma n^{2s})$~time to enumerate all maximal~$s$-plexes in~$G$. 
  Recall that~$a_n$ is the number of maximal $s$-plexes in a weakly~$\gamma$-closed graph on~$n$ vertices. 
  Thus, all maximal $s$-plexes of a weakly~$\gamma$-closed graph on~$n$ vertices can be enumerated in~$\Oh((a_n \cdot n + 2^\gamma n^{2s}) \cdot n) = \Oh(2^\gamma n^{2s + 1})$ time.
\end{proof}

A factor of~$n^{2s - 2}$ for the number of maximal~$s$-plexes in \Cref{theo-enumerate-max-s-plexes} is unavoidable:
Consider a graph $G$ consisting of two cliques~$C_1$ and~$C_2$ of equal size. 
Clearly, $G$ is~$1$-closed.
Each subset of~$C_1$ of size exactly~$s - 1$ and each subset of~$C_2$ of size exactly~$s - 1$ together form a maximal~$s$-plex. 
Hence, there exist~$1$-closed graphs with~$\Omega((n/2)^{2s - 2})$ maximal~$s$-plexes. 

For \textsc{$s$-Plex},
Theorem~\ref{theo-enumerate-max-s-plexes} directly implies the following.

\begin{corollary}
\label{cor:s-plex-in-fpt-time}
For~$s\ge 2$, \textsc{$s$-Plex} can be solved in $\Oh(2^{\gamma} n^{2s+1})$~time.
\end{corollary}

Next, we show that there is presumably no $f(k) \cdot n^{\Oh(1)}$-time algorithm for \textsc{$s$-Plex} in~$2$-closed graphs.
Moreover, our reduction also shows that \textsc{$s$-Plex} is W[1]-hard for the parameter $k + s+d$.

\begin{theorem}
  \label{thm-s-plex-w-hard}
  \textsc{$s$-Plex} is W[1]-hard in 2-closed graphs when parameterized by $k + s + d$.
\end{theorem}

\begin{proof}
  We reduce from \textsc{Clique}.
  An illustration of our construction is shown in Fig.~\ref{fig-example-s-plex}.
  Let $(G, k)$ be an instance of \textsc{Clique} with $k \ge 4$.
  First, we subdivide each edge $uv$ of $G$ twice.
  That is, we remove the edge $uv$ and add edges~$u x_{u}^v$, $x_{u}^v x_v^u$, and~$x_v^u v$, where $x_{u}^v$ and $x_v^u$ are two new vertices.
  Second, for each edge $uv \in E(G)$, we introduce~$k - 3$ vertices $x_{uv}^1, \dots, x_{uv}^{k - 3}$.
  Let $X_{uv} \coloneqq \{ x_{u}^v, x_v^u, x_{uv}^1, \dots, x_{uv}^{k - 3} \}$ and let $X \coloneqq \bigcup_{uv \in E(G)} X_{uv}$.
  We then add edges so that $X_{uv}$ forms a clique.
  Lastly, we introduce a set $T \coloneqq \{ t^1, \dots, t^{k - 3} \}$ of~$k - 3$~vertices and add edges between $x_{uv}^i$ and~$t^i$ for each $uv \in E(G)$ and each $i \in [k - 3]$.
  Let $G'$ be the resulting graph.
  
  \begin{figure}
  \centering
  \begin{tikzpicture}
  \node[label=below:{$u$}](gu) at (0, 0) [shape = circle, draw, fill=black, scale=0.07ex]{};
  \node[label=below:{$v$}](gv) at (1, 0) [shape = circle, draw, fill=black, scale=0.07ex]{};
  \node[label=above:{$w$}](gw) at (0.5, 1) [shape = circle, draw, fill=black, scale=0.07ex]{};
  \node[label=below:{$y$}](gy) at (2, 0) [shape = circle, draw, fill=black, scale=0.07ex]{};
  \path [-,line width=0.2mm](gu) edge (gv);
  \path [-,line width=0.2mm](gu) edge (gw);
  \path [-,line width=0.2mm](gw) edge (gv);
  \path [-,line width=0.2mm](gv) edge (gy);

    \node[label=below:{$u$}](u) at (5, 0) [shape = circle, draw, fill=black, scale=0.07ex]{};
  \node[label=below:{$v$}](v) at (7, 0) [shape = circle, draw, fill=black, scale=0.07ex]{};
  \node[label=above:{$w$}](w) at (6, 2) [shape = circle, draw, fill=black, scale=0.07ex]{};
  \node[label=below:{$y$}](y) at (9, 0) [shape = circle, draw, fill=black, scale=0.07ex]{};
    \path [-,line width=0.2mm](u) edge (v);
  \path [-,line width=0.2mm](u) edge (w);
  \path [-,line width=0.2mm](w) edge (v);
  \path [-,line width=0.2mm](v) edge (y);
  
    \node[label=below:{$x^v_u$}](xvu) at (5.6667, 0) [shape = circle, draw, fill=black, scale=0.07ex]{};
    \node[label=below:{$x^u_v$}](xuv) at (6.3333, 0) [shape = circle, draw, fill=black, scale=0.07ex]{};
    \node[label=below:{$x^y_v$}](xvu) at (7.6667, 0) [shape = circle, draw, fill=black, scale=0.07ex]{};
    \node[label=below:{$x^v_y$}](xuv) at (8.3333, 0) [shape = circle, draw, fill=black, scale=0.07ex]{};
    \node[label=left:{$x^w_u$}](xwu) at (5.3333, 0.6667) [shape = circle, draw, fill=black, scale=0.07ex]{};
    \node[label=left:{$x^u_w$}](xuw) at (5.6667, 1.3333) [shape = circle, draw, fill=black, scale=0.07ex]{};
    \node[label=right:{$x^v_w$}](xwv) at (6.3333, 1.3333) [shape = circle, draw, fill=black, scale=0.07ex]{};
    \node[label=right:{$x^w_v$}](xvw) at (6.6667, 0.6667) [shape = circle, draw, fill=black, scale=0.07ex]{};
    
        \node[label=above:{$x^1_{uv}$}](x1uv) at (5.6, -1.2) [shape = circle, draw, fill=black, scale=0.07ex]{};
        \node[label=above:{$x^{k-3}_{uv}$}](xkuv) at (6.4, -1.2) [shape = circle, draw, fill=black, scale=0.07ex]{};
        \draw (6,-0.6) ellipse (0.9cm and 0.9cm);
        \node[label=left:{$X_{uv}$}](xuv) at (5.3, -0.9) {};
        
        \node[label=above:{$x^1_{vy}$}](x1yv) at (7.6, -1.2) [shape = circle, draw, fill=black, scale=0.07ex]{};
        \node[label=above:{$x^{k-3}_{vy}$}](xkyv) at (8.4, -1.2) [shape = circle, draw, fill=black, scale=0.07ex]{};
        \draw (8,-0.6) ellipse (0.9cm and 0.9cm);
        \node[label=right:{$X_{vy}$}](xyv) at (8.7, -0.9) {};
  
\node[label=below:{$t^1$}](t1) at (6.6, -2) [shape = circle, draw, fill=black, scale=0.07ex]{};
\node[label=below:{$t^{k-3}$}](tk) at (7.4, -2) [shape = circle, draw, fill=black, scale=0.07ex]{};
\draw (7,-2) ellipse (1.2cm and 0.55cm);
\path [-,line width=0.2mm](t1) edge (x1uv);
\path [-,line width=0.2mm](t1) edge (x1yv);
\path [-,line width=0.2mm](tk) edge (xkuv);
\path [-,line width=0.2mm](tk) edge (xkyv);
\node[label=right:{$T$}](xyv) at (8, -2) {};

\node[label=left:{$a)$}](xyv) at (0, 2.5) {};
\node[label=left:{$b)$}](xyv) at (5, 2.5) {};
  \end{tikzpicture}
  \label{fig-example-s-plex}
  \caption{Illustration of the construction of Theorem~\ref{thm-s-plex-w-hard}. $a)$ shows the graph~$G$ of the \textsc{Clique} instance and~$b)$ shows the graph~$G'$ of \textsc{$s$-Plex}. 
  Here, the sets~$X_{uw}$ and~$X_{vw}$ are not drawn.
  Note that the sets~$X_{uv}$ and~$X_{vy}$ are cliques.}
\end{figure}

  It is easy to verify that $G'$ is 2-closed.
  Moreover, $G'$ is $(k - 1)$-degenerate:
  Each vertex $x \in X_{uv}$ is of degree $k - 1$ and there is no edge in $G' - X$.
  We show that $G$ has a clique of size $k$ if and only if $G'$ has an $s$-plex of size $k'$, where $k' \coloneqq 2k - 3 + (k - 1) \binom{k}{2}$ and $s \coloneqq k' - (k - 1)$.

  Suppose that $G$ has a clique $S$ of size exactly $k$.
  Let $S' = S \cup T \cup \bigcup_{u,v \in S} X_{uv}$.
  Observe that $ \vert S' \vert  = k'$.
  We verify that each vertex in $G'[S']$ has degree at least $k' - s = k - 1$.
  \begin{itemize}
    \item
      Let $v \in S$.
      By construction, we have $x_{v}^u \in N_{G'}(v)$ for each $u \in S \setminus \{ v \}$.
      Since $x_{v}^u$ is contained  in $S'$, $v$ has at least $k - 1$ neighbors in $G'[S']$.
    \item
      We have $\deg_{G'[S']}(t^i) \ge \binom{k}{2} \ge k - 1$ for each $i \in [k - 3]$, because $t^i$ is adjacent to $x_{uv}^i$ for all $uv \in E(G[S])$.
    \item
      Consider $x_{u}^v$ for $uv \in E(G[S])$.
      We have $u \in N_{G'}(x_{u}^v)$ by construction.
      Moreover, $x_{u}^{v}$ is adjacent to all $k - 2$ vertices in $X_{uv} \setminus \{ x_{u}^v \}$.
      Thus, we have $\deg_{G'[S']}(x_{u}^v) \ge k - 1$.
    \item
      Consider $x_{uv}^i$ for $uv \in E(G[S])$ and $i \in [k - 3]$.
      We have $t^i \in N_{G'}(x_{uv}^i)$ by construction.
      Moreover, $x_{uv}^i$ is adjacent to all $k - 2$ vertices in $X_{uv} \setminus \{ x_{uv}^i \}$.
      Thus, we have $\deg_{G'[S']}(x^i_{uv}) \ge k - 1$.
  \end{itemize}
  Thus, every vertex has at least $k - 1 = k' - s$ neighbors in $G'[S']$.

  Conversely, suppose that $S'$ is an $s$-plex of size exactly $k'$.
  We start with the following claim.
  \begin{claim}
    If $S'$ contains a vertex $x$ of $X_{uv}$ for some $uv \in E(G)$, then $S'$ also contains all vertices in $N_{G'}[X_{uv}]$, that is, $\{u, v\}\cup X_{uv}\cup T \subseteq S'$.
  \end{claim}
  \begin{claimproof}
    By construction, $\deg_{G'}(x) = k - 1$.
    Since each vertex in $G'[S']$ has degree~$ \vert S' \vert  - s \ge k - 1$ by the definition of $s$-plexes, we have $N_{G'}[X_{uv}] \subseteq S'$.
    %The claim follows because $X_{uv}$ is a clique.
  \end{claimproof}

  Let $\ell =  \vert S' \cap V(G) \vert $.
We conclude that there are at most $\binom{\ell}{2}$ edges $uv \in E(G)$ with $X_{uv} \cap S' \ne \emptyset$ since otherwise the above claim would imply that~$\vert S' \cap V(G) \vert>\ell$.
  By construction, we have $ \vert X_{uv} \vert  = k - 1$ for each $uv \in E(G)$.
  Thus, we have
  \begin{align*}
     \vert S' \vert  =  \vert S' \cap V(G) \vert  +  \vert T \vert  +  \vert S' \cap X \vert  \le \ell + k - 3 + (k - 1) \binom{\ell}{2}.
  \end{align*}
  Since $ \vert S' \vert  = k' = 2k - 3 + \binom{k}{2}$, we obtain $\ell \ge k$.

  By definition, each vertex $v \in S' \cap V(G)$ has at least $ \vert S' \vert  - s \ge k - 1$ neighbors in $G'[S']$. 
  So there are at least $\ell (k - 1) / 2$ edges $uv \in E(G)$ such that $S' \cap X_{uv} \ne \emptyset$.
  From the above claim we know that~$X_{uv}\subseteq S'$ for each~$X_{uv}$ with~$X_{uv}\cap S'\ne\emptyset$.
  Hence, we obtain that
  \begin{align*}
     \vert S' \vert  \ge  \vert S' \cap V(G) \vert  +  \vert T \vert  +  \vert S' \cap X \vert  \ge \ell + k - 3 + (k - 1) \cdot \ell (k - 1) / 2.
  \end{align*}
  Since $ \vert S' \vert  = k' = 2k - 3 + (k - 1) \binom{k}{2}$, we obtain $\ell = k$ and $ \vert S' \cap X \vert  = (k - 1) \binom{k}{2}$.
  Since~$\vert S' \cap X \vert  = (k - 1) \binom{k}{2}$ we conclude that each two vertices in~$S'\cap V(G)$ are adjacent.
  Thus,~$S' \cap V(G)$ is a clique of $k$ vertices in $G$ by construction.
\end{proof}

\subsection{\texorpdfstring{$s$}{s}-Defective Clique}

A clique is a vertex set~$S$ such that there exists no vertex pair in~$S$ which is nonadjacent. 
One way to relax the clique definition is to allow up to~$s$ nonadjacent vertex pairs. 
This idea can be formalized as follows.

\begin{definition}
In a graph~$G=(V,E)$ a set~$S\subseteq V$ is an \emph{$s$-defective clique} if~$G[S]$ has at least~$\binom{ \vert S \vert }{2} - s$ edges.
\end{definition}

Note that cliques are exactly the 0-defective cliques.
Similar to $s$-plexes, we consider the problems of enumerating all maximal $s$-defective cliques and finding a sufficiently large $s$-defective clique in (weakly) closed graphs.
The latter problem can be formalized as follows.

\problemdef{$s$-Defective Clique}
{A graph~$G$ and~$k \in \mathds{N}$.}
{Does~$G$ contain an~$s$-defective clique~$S$ of size at least $k$?}

One can show that \textsc{$s$-Defective Clique} is W[1]-hard with respect to~$k$ even if $c = 2$ by adapting a previous hardness proof for \textsc{Densest-$k$-Subgraph} on~$2$-degenerate graphs~\cite[Theorem 20]{RS08}.

First, we study the problem of enumerating all maximal $s$-defective cliques.
To obtain an FPT-algorithm for this problem for the weak closure number, we adapt the algorithm of \Cref{theo-enumerate-max-s-plexes}.
The only difference to the proof of Theorem~\ref{theo-enumerate-max-s-plexes} is the following:
For bounding the number of~$s$-plexes of Type~$3$ the sets~$\widetilde{S_v}$ and~$\widetilde{S_{uv}}$ were bounded by~$s-1$ each. 
Since a maximal~$s$-defective clique contains at most~$s$ non-edges and~$uv\notin E(G)$ we observe that~$ \vert \widetilde{S_v}\cup\widetilde{S_{uv}} \vert <s$. 
Hence, there are at most~$2^{\gamma-1}n^{s-1}$ maximal~$s$-defective cliques of Type~$3$. 
Thus, we can bound the overall number of maximal~$s$-defective cliques by~$2^{\gamma}n^{s+1}+1$.  
Since the rest of the proof is completely analogous, we omit it.

\begin{theorem}
  \label{thm:defect-clique}
  For~$s\ge 2$, there are~$\Oh(2^{\gamma}n^{s+1})$ maximal~$s$-defective cliques in weakly~$\gamma$-closed graphs and they can be enumerated in~$\Oh(2^\gamma n^{s+3})$~time.
\end{theorem}

A factor of~$n^{s + 1}$ in the number of maximal~$s$-defective cliques in \Cref{thm:defect-clique} is inevitable due to the following lower bound:
Again we consider the graph~$G$ consisting of two disjoint cliques $C_1$ and $C_2$, each of size $n/2$.
For each clique~$C \subseteq C_1$ of size $s$ and each~$v \in C_2$, the vertex set~$C \cup \{ v \}$ is a maximal $s$-defective clique.
Thus, $G$ has $\Omega((n/2)^{s+1})$ maximal $s$-defective cliques.

Second, we study~\textsc{$s$-Defective Clique}, the decision problem of finding a sufficiently large $s$-defective clique. Theorem~\ref{thm:defect-clique} directly implies the following.

\begin{corollary}
\label{cor-s-def-clique-gamma}
\textsc{$s$-Defective Clique} can be solved in~$\Oh(2^\gamma n^{s+3})$~time.
\end{corollary}

Next, we present faster algorithms in terms of the dependence on~$s$. 
First, we show that each~$s$-defective clique can be covered by $\Oh(\sqrt{s})$ maximal cliques. 

\begin{lemma}
  \label{lemma:cc}
  Let $S$ be an $s$-defective clique for $s \ge 1$.
  Then, there is a collection $\mathcal{C}$ of at most $\Oh(\sqrt{s})$ cliques such that $S \subseteq \bigcup_{C \in \mathcal{C}} C$.
\end{lemma}
\begin{proof}
Let~$H$ denote the complement graph of~$G[S]$.
  By definition,~$H$ has at most $s$~edges.
  Since a clique becomes an independent set in the complement graph, it suffices to show that there is an $\Oh(\sqrt{s})$-coloring of $H$ (that is,~$\chi(H) = \Oh(\sqrt{s})$).
  Although this is known folklore, we describe its proof for the sake of completeness.
  Consider an optimal coloring.
  Then, for each pair of colors, say red and blue, there is at least one edge with one endpoint red and the other blue (otherwise we find a coloring with fewer colors).
  Recall that~$H$ is the complement graph of~$G[S]$. 
  Hence,~$H$ has at most~$s$ edges, we obtain~$s \ge \binom{\chi(H)}{2}$, or equivalently, $\chi(H) \le\sqrt{2s + \frac{1}{4}} + \frac{1}{2}$.
\end{proof}

Note that a trivial brute-force algorithm can enumerate all (not necessarily maximal) cliques in $\Oh(2^d d n)$~time.
Lemma~\ref{lemma:cc} says that each~$s$-defective clique is covered by at most $\Oh(\sqrt{s})$~cliques.
Hence, by a simple brute-force we obtain the following.

\begin{theorem}
  \textsc{$s$-Defective Clique} can be solved in $2^{\Oh(d \sqrt{s})} n^{\Oh(\sqrt{s})}$ time.
\end{theorem}

We can also use Lemma~\ref{lemma:cc} to obtain an algorithm in terms of the smaller parameter~$\gamma$ instead of the degeneracy~$d$ without increasing the exponent of~$n$.

\begin{theorem}
\label{thm:s-def-clique-gamma}
  \textsc{$s$-Defective Clique} can be solved in $2^{\Oh(\gamma \sqrt{s} + s \log k}) n^{\Oh(\sqrt{s})}$~time.
\end{theorem}
\begin{proof}
  We first enumerate all maximal cliques in $(3^{\gamma / 3} \cdot n^{\Oh(1)})$~time \cite{FRSWW20}.
  If there is a clique of size at least $k$, then return Yes, since each clique is also an~$s$-defective clique.
  Now, we assume that there is no clique of size at least~$k$.
  By \Cref{lemma:cc}, it suffices to check whether there is an $s$-defective clique of size $k$ in $\bigcup_{C \in \mathcal{C}} C$ for each collection~$\mathcal{C}$ of $\Oh(\sqrt{s})$ maximal cliques.
  Observe that each fixed collection in~$\mathcal{C}$ has~$\Oh(k \sqrt{s})$ vertices.
  Let~$W_{\mathcal{C}}$ denote the vertex set of~$\mathcal{C}$.
  By applying the algorithm of Corollary~\ref{cor-s-def-clique-gamma} to find the largest~$s$-defective clique, we can determine in $\Oh(2^\gamma(\sqrt{s}k)^{\Oh(s+3)})$~time whether~$W_{\mathcal{C}}$ contains an~$s$-defective clique of size at least~$k$.
  Since there are $\Oh^*(3^{\gamma / 3})$ maximal cliques, the overall running time of this algorithm is~$(3^{\gamma / 3} \cdot n^{\Oh(1)})^{\Oh(\sqrt{s})} \cdot \Oh(2^\gamma(\sqrt{s} k)^{\Oh(s+3)}) = 2^{\Oh(\gamma \sqrt{s} + s \log k)} n^{\Oh(\sqrt{s})}$~time.
\end{proof}

For $c$-closed graphs, we can obtain an algorithm whose running time does not depend on~$k$.
This is due to the following lemma.

\begin{lemma}
  \label{lemma:defclsize}
  Let $S \subseteq V(G)$ be an $s$-defective clique in $G$, in which at least one pair of vertices is nonadjacent.
  Then, $ \vert S \vert  \le c + s$.
\end{lemma}
\begin{proof}
  Let $u, v \in S$ be vertices such that $uv \notin E(G)$.
  We show that $ \vert S' \vert  \le c + s - 2$ for~$S' \coloneqq  S \setminus \{ u, v \}$.
  Since $G$ is $c$-closed, there are at most $c - 1$ vertices in $S'$ adjacent to both $u$ and $v$.
  Moreover, there are at most $s - 1$ vertices in $S'$ which are nonadjacent to either $u$ or $v$ in~$S'$, by the definition of $s$-defective cliques.
  Thus, we obtain $ \vert S' \vert  \le (c - 1) + (s - 1) = c + s - 2$.
\end{proof}

From Lemma~\ref{lemma:defclsize} we directly obtain the following.

\begin{corollary}
  \textsc{$s$-Defective Clique} can be solved in $2^{\Oh(c \sqrt{s} + s \log (c + s))} n^{\Oh(\sqrt{s})}$ time.
\end{corollary}

\subsection{\texorpdfstring{$2$}{2}-Clubs}

A clique is a vertex set~$S$ such that each vertex in~$S$ has distance~$1$ to each other vertex within~$S$. 
One way to relax the clique definition is to allow greater distances of pairs of vertices within~$G[S]$. 
This idea can be formalized as follows.

\begin{definition}
\label{def-s-club}
For a graph~$G$ and $s \in \mathds{N}$, a set~$S\subseteq V(G)$ is an \emph{$s$-club} if each pair of vertices in~$S$ has distance at most~$s$ in~$G[S]$.
\end{definition}

Note that cliques are exactly the 1-clubs.
This definition leads to the following decision problem.

\problemdef{$2$-Club}
{A graph~$G$ and~$k \in \mathds{N}$.}
{Does~$G$ contain a~$2$-club~$S$ of size at least k?}

It is known that \textsc{$2$-Club} parameterized by~$k$ admits an FPT-algorithm~\cite{CHLS13,SKMN12} and that it does not admit a polynomial kernel unless \PHC~\cite{SKMN12}.
Since the largest $2$-club containing some vertex~$v$ is~$N_2[v]$, we observe that the size of a largest $2$-club is~$\Delta^2+1$.
This implies fixed-parameter tractability for~$\Delta$.
In contrast, \textsc{$2$-Club} is W[1]-hard with respect to~$h$-index and it is NP-hard on 6-degenerate graphs~\cite{HKNS15}.
Since~$\gamma\le d+1$, this also implies NP-hardness for constant values of~$\gamma$.
We extent these results, by showing that \textsc{$2$-Club} remains NP-hard even on 4-closed graphs.

\begin{theorem}
\label{thm-2-club-nphard-constant-closure}
\textsc{$2$-Club} remains NP-hard even on~$4$-closed graphs.
\end{theorem}
\begin{proof}
We reduce from \textsc{Clique}. 

\subparagraph*{Construction.}
Let~$(G,k)$ be an instance of \textsc{Clique}.
We construct an equivalent instance~$(G',k')$ of \textsc{$2$-Club} such that~$G'$ is~$4$-closed.
We set~$k'\coloneqq k\cdot n^2$.
For each vertex~$w\in V(G)$, we add a clique~$K_w\coloneqq \{w_j\mid j\in \{ 0,\dots, n^2-1 \} \}$ of size~$n^2$ to~$G'$.
We denote the graph constructed so far by~$G^0$.
Furthermore, let~$(e_1,e_2, \ldots, e_m)$ be an arbitrary but fixed ordering of the edges in~$E(G)$.
We will add edges corresponding to each edge~$e_i\in E(G)$ to~$G'$.
We denote by $G^i$ the graph after we added the gadgets for the edges~$e_1$ to~$e_i$ to~$G^0$.
Note that~$G^0$ is the graph constructed so far; a disjoint union of cliques, and that~$G^m=G'$.
The idea for the gadget of edge~$e_i=uv$ is as follows:
We add a matching between the vertices of the cliques~$K_u$ and~$K_v$. 
More precisely, we add the edges~$u_iv_{i+\ell_{uv}\bmod n^2}$ for each~$i\in\{ 0,\dots, n^2-1 \}$ and some fixed integer~$\ell_{uv}$.
We call~$\ell_{uv}$ the \emph{shift} of~$uv$.
We will assume that $\ell_{uv} + \ell_{vu} = n^2$.
To simplify notation, we will assume that the modulo $n^2$ is taken after the addition of a shift.
The difficult part lies in choosing~$\ell_{uv}$ carefully to obtain a graph with constant closure.

For a vertex pair~$(a, b)$ of~$G$ by~$A_{ab}$ we denote the set of vertices in the cliques~$K_a$ and~$K_b$ and by~$B_{ab}$ the remaining vertices of~$V(G')$.
Next, we prove the following invariant which is an essential ingredient to show that~$G'=G^m$ has constant closure number.
% Roughly speaking, the following invariant guarantees us that there exists a sequence of graphs~$G^0,\ldots, G^m$ such that each of these graphs has constant $c$-closure.

\begin{quote}
\textit{Invariant.} For each~$i$, there is a shift~$\ell_{uv}$ for the gadget of the $i$th~edge~$e_i=uv$ such that for each two nonadjacent vertices~$x\in K_a$ and~$y\in K_b$ for any vertices~$a,b\in V(G)$ we have~$ \vert N_{G^i}(x)\cap N_{G^i}(y)\cap B_{ab} \vert \le 1$.
Moreover, we can find $\ell_{uv}$ is polynomial time.
\end{quote}

That is, we want to maintain the invariant that two nonadjacent vertices in $K_a \cup K_b$ have at most one common neighbor in $B_{ab}$.
% Observe that choosing~$x$ and~$y$ from different vertex gadgets includes all possible cases for nonadjacent vertices in~$G^i$ since each vertex gadget is a clique.
Recall that~$G^0$ is a disjoint union of cliques. Thus, the invariant holds for~$G^0$.
In the following, we assume that the invariant holds for the graph~$G^{i-1}$.
Recall that the graph~$G^i$ is constructed from~$G^{i-1}$ by adding the matching for the edge~$e_i= uv$.
We will show that the invariant can be maintained for~$G^i$.
More precisely, we show that we can compute a shift~$\ell_{uv}\in\{ 0,\dots, n^2-1 \}$ in polynomial time such that adding the edges~$u_jv_{j+\ell_{uv}}$ for each~$j\in\{ 0,\dots, n^2-1 \}$ to~$G^{i-1}$ does not violate the invariant.

% Note that since each vertex in~$G^i$ has at most one more neighbor than in~$G^{i-1}$, we directly obtain that~$ \vert N_{G^i}(x)\cap N_{G^i}(y)\cap B_{uv} \vert \le 2$ for any vertices~$u,v\in V(G)$ and any two nonadjacent vertices~$x,y\in A_{uv}$.
% Thus, the task is to show that the vertices~$x$ and~$y$ also have at most one common neighbor in~$G^i$.

Assume to the contrary that the invariant is violated by two nonadjacent vertices in~$G^i$. 
Observe that there could be three possibilities on how the invariant could be violated in~$G^i$:
\begin{enumerate}[label=Case \arabic*:,leftmargin=*]
\item Two nonadjacent vertices in~$A_{pq}$ for~$p,q\in V(G)\setminus\{u,v\}$ violate the invariant,
\item two nonadjacent vertices in~$A_{uv}$ violate the invariant, or
\item two nonadjacent vertices in~$A_{wp}$ for~$w\in\{u,v\}$ and~$p\in V(G)\setminus\{u,v\}$ violate the invariant.
\end{enumerate}

In the following, we show that we can choose the shift~$\ell_{uv}$ in such a way to fulfill the invariant also for~$G^i$.
We distinguish the three above cases:

\subparagraph*{Case~$1$.} Let~$x$ and~$y$ be a pair of nonadjacent vertices in~$A_{pq}$ violating the invariant.
Note that each edge added to~$G^{i-1}$ to obtain~$G^i$ is of the form~$u_rv_s$.
Clearly,~$u_r,v_s\notin A_{pq}$.
Hence, we conclude that~$ \vert N_{G^i}(x)\cap N_{G^i}(y)\cap B_{pq} \vert = \vert N_{G^{i-1}}(x)\cap N_{G^{i-1}}(y)\cap B_{pq} \vert \le 1$ since the invariant holds for~$G^{i-1}$.
Thus, this case is not possible.

\subparagraph*{Case~$2$.} Let~$x$ and~$y$ be a pair of nonadjacent vertices in~$A_{uv}$ violating the invariant.
As in case~$1$, since only edges with both endpoints in~$A_{uv}$ are added to the graph~$G'$, we obtain that~$ \vert N_{G^i}(x)\cap N_{G^i}(y)\cap B_{uv} \vert = \vert N_{G^{i-1}}(x)\cap N_{G^{i-1}}(y)\cap B_{uv} \vert \le 1$ since the invariant holds for~$G^{i-1}$.
Thus, this case is also not possible.

\subparagraph*{Case~$3$.} Without loss of generality, assume that~$w=u$.
Recall that adding a matching between the cliques~$K_u$ and~$K_v$ can increase the number of common neighbors in~$B_{up}$ of two nonadjacent vertices in~$A_{up}$ by at most~$1$.
Thus, two vertices in~$A_{up}$ violating the invariant in~$G^i$ have a common neighbor in some clique~$K_t$ in~$G^{i-1}$.
Since only the matchings corresponding to the edges~$ut,pt\in E(G)$ result in edges between~$K_u$ and~$K_t$ and between~$K_p$ and~$K_t$, the matchings corresponding to the edges~$ut$ and~$pt$ are already added to~$G^{i-1}$.
To obtain~$G^i$ from~$G^{i-1}$ only a matching between~$K_u$ and~$K_v$ is added.
Thus, we conclude that the matching corresponding to the edge~$pv$ was already present in~$G^{i-1}$.

%Now, we obtain the following: 
%For each vertex~$x\coloneqq u_\alpha\in K_u$ we denote by~$t_\beta\in K_t$ the unique neighbor of~$u_\alpha$ in~$K_t$ determined by the shift~$\ell_{ut}$.
%Furthermore, we denote by~$y=p_\gamma\in K_p$ the unique neighbor of~$t_\beta$ in~$K_p$ determined by the shift~$\ell_{pt}$.
%Next, we denote by~$w_\delta\in K_w$ the unique neighbor of~$p_\gamma$ in~$K_w$ determined by the shift~$\ell_{pw}$.

For every $j \in \{ 0,\dots, n^2-1 \}$, we have $N(u_j) \cap K_t = \{  t_{j + \ell_{ut}} \}$ and $N(t_{j + \ell_{ut}}) \cap K_p = \{ p_{j + \ell_{ut} + \ell_{tp}} \}$.
Hence, $u_j$ and $p_{j'}$ have a common neighbor in $K_p$ if and only if $j' - j \equiv \ell_{ut} + \ell_{tp}$.
Similarly, $u_j$ and $p_{j'}$ have a common neighbor in $K_v$ if and only if $j'- j \equiv \ell_{uv} + \ell_{vp}$.
Consequently, the invariant is only violated if $\ell_{uv} \equiv \ell_{ut} + \ell_{tp} +\ell_{pv}$.
%Hence, the only forbidden shift is~$\ell_{uw}\coloneqq \delta-\alpha$ independent of the initial choice of~$x$.
Thus, for each $p$ and $t$, there is at most one shift violating the invariant, amounting to at most $(n - 2)^2$ forbidden shifts.
Since there are~$n^2$ possible shifts, we conclude that we can choose a shift~$\ell_{uv}$ in a way which does not violate the invariant.
Note that this does not only show the existence of a shift maintaining the invariant, the above argument also shows that the shift~$\ell_{uv}$ can be constructed in polynomial time, although no explicit formula for~$\ell_{uv}$ is given here.

Thus, we have shown that the invariant is maintained for each~$i$, in particular for~$i=m$ and hence for the resulting graph~$G'$.

\subparagraph*{Bounded Closure.}
We use the invariant to show that~$G'$ is~$4$-closed.
Consider two nonadjacent vertices~$x\in K_u$ and~$y\in K_v$ in~$G'$.
Observe that~$u\ne v$ since otherwise~$xy\in E(G')$.
By the invariant, we have~$ \vert N_{G'}(x)\cap N_{G'}(y)\cap B_{uv} \vert \le 1$.
Recall that~$A_{uv}=K_u\cup K_v$.
Since~$x$ has at most one neighbor in~$K_v$ and since~$y$ has at most one neighbor in~$K_u$, we conclude that~$x$ and~$y$ have at most three common neighbors.
Thus,~$G'$ is~$4$-closed.

\subparagraph*{Correctness.}
Suppose that~$G$ contains a clique~$C$ of size at least~$k$.
Let~$S\coloneqq \{K_v\mid v\in C\}$.
Clearly,~$S$ has size~$k'=k\cdot n^2$.
It remains to show that~$S$ is a~$2$-club.
Consider two  nonadjacent vertices~$x,y\in S$.
Note that~$x\in K_u$ and~$y\in K_v$ for~$u,v\in C$ such that~$u\ne v$ since otherwise~$xy\in E(G')$.
Since~$C$ is a clique, we have~$uv\in E(G)$ and thus we added a matching between the cliques~$K_u$ and~$K_v$.
Hence,~$x$ has a neighbor~$z$ in~$K_v$ and thus~$x$ and~$y$ have distance~$2$ since~$K_v$ is a clique.

Conversely, suppose that~$S$ contains an~$2$-club~$S$ of size at least~$k'=k\cdot n^2$. 
Let~$T\coloneqq \{v\mid  \vert K_v\cap S \vert \ge n+1\}$.
Observe that~$ \vert T \vert \ge k$, since otherwise $\vert S \vert \le \vert T \vert \cdot n^2+(n-\vert T \vert)\cdot n \le kn^2 - (k - 1)n$.
In the following, we show that~$T$ is a clique in~$G$.
Assume towards a contradiction that~$T$ is not a clique and let~$u,v\in T$ such that~$uv\notin E(G)$.
Let~$u^*$ be a vertex in~$K_u\cap S$ and let~$U\coloneqq N(u^*)\setminus K_u$.
Note that since~$uv\notin E(G)$ we have~$U\cap K_v=\emptyset$.
Furthermore, note that by construction each vertex~$y\in K_w$ has at most one neighbor in~$K_x$ for any~$w,x\in V(G)$ such that~$w\ne x$.
Thus,~$ \vert U \vert \le n$.
Furthermore, by the same argument we obtain that each vertex in~$U$ has at most $1$~neighbor in~$K_v$.
Thus,~$u^*$ has distance at most~$2$ to at most~$n$ vertices in~$K_v$.
This is a contradiction to the fact that~$ \vert S\cap K_v \vert \ge n+1$ and that~$S$ is an~$2$-club.
Hence,~$T$ is a clique and thus~$G$ contains a clique of size at least~$k$.
\end{proof}

We leave the complexity of \textsc{2-Club} on 2-closed graphs and 3-closed graphs open.
We want to point out that 2-closed graphs of diameter two are also known to be \emph{geodetic}, that is, each pair of vertices has a unique shortest path between them.
Moreover, it is known that every 2-closed graph $G$ of diameter two satisfies one of the following \cite{BB88}:
\begin{itemize}
  \item
    $G$ contains a vertex $v$ such that $N(v) = V(G)$, or
  \item
    $G$ is strongly regular, that is, $G$ is regular and for some~$\lambda, \mu \in \mathds{N}$, every two adjacent (nonadjacent) vertices have $\lambda$ ($\mu$, respectively) common neighbors (note that $\mu = 1$ since $G$ is 2-closed), or
  \item
    $G$ has exactly two vertex degrees.
\end{itemize}

To show that \textsc{$2$-Club} in $2$-closed graph is solvable in polynomial time exploiting these three properties might be helpful.

For $2$-clubs we only studied the decision variant \textsc{$2$-Club} in which we ask for an sufficiently large $2$-club in~$c$-closed graphs.
The enumeration of all maximal $2$-clubs is \emph{not} possible in FPT-time even for graphs with constant closure:
Observe that in the construction of Theorem~\ref{thm-2-club-nphard-constant-closure},~$C$ is a maximal clique in the graph~$G$ of the \textsc{Clique} instance if and only if~$\{K_c\mid c\in C\}$ is a maximal $2$-club in the graph~$G'$ of the \textsc{$2$-Club} instance.
The number of maximal cliques in an~$n$-vertex graph is~$3^{n/3}$~\cite{moonM1965}. 
Hence, the above correspondence shows that even a 4-closed graph may have up to $3^{n/3}$~maximal $2$-clubs.

%\begin{theorem}
%  \textsc{$s$-Defective Clique} can be solved in $2^{\Oh(c \sqrt{s} + (c + s) \log (c + s)} n^{\Oh(\sqrt{s})}$ time.
%\end{theorem}
%\begin{proof}
%  We first enumerate all maximal cliques in $\Oh^*(3^{c / 3})$ time \cite{FRSWW20}.
%  If there is a clique of size at least $k$, then return Yes.
%  If there is no clique of size $k$ and $k > c + s$, then return No (\Cref{lemma:defclsize}).
%  Now we assume that there is no clique of size $k \le c + s$.
%  By \Cref{lemma:cc}, it suffices to check whether there is an $s$-defective clique of size $k$ in $\bigcup_{C \in \mathcal{C}} C$ for each collection $\mathcal{C}$ of $\Oh(\sqrt{s})$ maximal cliques. 
%  Since there are $3^{c / 3} \cdot n^{\Oh(1)}$ maximal cliques, this procedure requires $(3^{c / 3} \cdot n^{\Oh(1)})^{\Oh(\sqrt{s})} \cdot \Oh(\sqrt{s} k)^{k}) = 2^{\Oh(c \sqrt{s} + (c + s) \log (c + s))} n^{\Oh(\sqrt{s})}$ time.
%\end{proof}

\section{Bicliques}

The counterpart of cliques in bipartite graphs are (non-) induced bicliques.
In this section we study the parameterized complexity of enumerating all maximal (non-) induced bicliques and finding a sufficiently large (non-) induced biclique in (weakly) closed graphs.

\subsection{Non-Induced Biclique}

In this subsection, we study problems of finding non-induced maximal bicliques fulfilling certain cardinality constraints. 
Next, we formally define non-induced bicliques.

\begin{definition}
In a graph~$G=(V,E)$ two disjoint vertex sets~$S\subseteq V$ and~$T\subseteq V$ are a \emph{non-induced biclique} if~$st\in E(G)$ for each~$s\in S$ and each~$t\in T$.
\end{definition}

There is an algorithm that enumerates
in~$\Oh^*(2^{d})$~time all maximal pairs of sets~$S$ and~$T$ such that each vertex of~$S$
is adjacent to each vertex of~$T$~\cite{Epp94}.\footnote{Eppstein~\cite{Epp94} describes an algorithm with running time~$\Oh^*(2^{2a})$ for the graph parameter arboricity~$a$ which is linearly bounded in~$d$ by the inequality~$a\le d\le 2a-1$. It can be shown that this algorithm also has running time~$\Oh^*(2^d)$.}
We also consider the problem of finding a sufficiently large non-induced biclique.

\problemdef
{Non-Induced~$(k_1, k_2)$-Biclique}
{A graph $G$ and $k_1, k_2 \in \mathds{N}$.}
{Does~$G$ contain a non-induced biclique with vertex sets~$S$ and~$T$ such that~$\vert S\vert \ge k_1$ and~$\vert T\vert \ge k_2$?}

% \problemdef
% {Non-Induced~$k$-Edge Biclique}
% {A graph $G$ and $k \in \mathds{N}$.}
% {Are there two disjoint sets $S, T$ such that~$ \vert S \vert \cdot \vert T \vert \ge k$ and $s t \in E(G)$ for each $s \in S$ and $t \in T$?}
\textsc{Non-Induced~$(k_1, k_2)$-Biclique} is W[1]-hard with respect to~$k_1$ even if~$k_1=k_2$~\cite{L18}.
We also consider \textsc{Non-Induced Max-Edge Biclique} where we demand
that~$ \vert S \vert \cdot  \vert T \vert \ge k$ instead of putting constraints on the partition
sizes.
We may assume that $\min \{ \vert S \vert, \vert T\vert \} \le \sqrt{k}$.
Thus, \textsc{Non-Induced Max-Edge Biclique} can be solved by solving~$\sqrt{k}$
instances of \textsc{Non-Induced~$(k_1, k_2)$-Biclique} and thus the latter problem can be
considered to be more difficult in our setting. \textsc{Non-Induced Max-Edge Biclique} can
be solved in $\Oh(k^{2.5}k^{\sqrt{k}}n)$~time by applying the algorithm for
\textsc{Induced Max-Edge Biclique} on bipartite graphs~\cite{FLZW18}.

First, we study the parameterized complexity of enumerating all maximal non-induced bicliques in weakly~$\gamma$-closed graphs.
 We need to define carefully, however, what we mean by
enumerating bicliques: 
The algorithm of Eppstein~\cite{Epp94} enumerates
in~$\Oh^*(2^{d})$~time all maximal pairs of sets~$S$ and~$T$ such that each vertex of~$S$
is adjacent to each vertex of~$T$.
For this enumeration problem, an FPT-algorithm for the weak closure is unattainable since any clique of size~$n$ is 1-closed and
admits $\Theta(2^{n})$~bipartitions that need to be enumerated. To circumvent this
issue, we view a biclique as a vertex set that can be partitioned into sets~$S$
and~$T$. Thus, in order to strengthen the parameterization from~$d$ to~$\gamma$, we go from an explicit listing of bicliques with bipartitions to a compact representation of bicliques as vertex sets and this is indeed necessary.
% In contrast to the hardness result for \textsc{Induced~$(k_1, k_2)$-Biclique} for~$k_1\neq k_2$ in weakly~$k_1+2$-closed graphs, we next show that \textsc{Non-Induced~$(k_1, k_2)$-Biclique} is FPT parameterized by~$\gamma$. 
We say that a vertex set $U \subseteq V(G)$ is a non-induced biclique if $G[U]$ contains a biclique as a (not necessarily induced) subgraph.
Note that it can be decided in $\Oh(n^2)$ time whether a vertex set $U \subseteq V(G)$ is a non-induced biclique or not, because~$U$ is a non-induced biclique if and only if the complement of $G[U]$ has multiple connected components.
We adapt the algorithm of \Cref{theo-enumerate-max-s-plexes} to obtain an $\Oh^*(2^\gamma)$-time algorithm to enumerate all maximal non-induced bicliques.

Recall that in Theorem~\ref{theo-enumerate-max-s-plexes} we bounded the overall number of maximal $s$-plexes in a weakly~$\gamma$-closed graph~$G$ by distinguishing 4 different types of maximal $s$-plexes if we are provided wih the set of maximal $s$-plexes of~$G-v$.
As in the proof of \Cref{theo-enumerate-max-s-plexes}, we aim to enumerate all maximal non-induced bicliques in~$G$, provided with the collection $\mathcal{S'}$ of all non-induced maximal bicliques in $G' \coloneqq  G - v$.
Again, we define the same four types of non-induced bicliques~$S$:
Type~$1$:  $S$ does not contain $v$, Type~$2$: $S$ contains $v$ and $S \setminus \{ v \}$ is maximal in $G'$, Type~$3$: $S$ contains $v$,~$S \setminus \{ v \}$ is not maximal in $G'$, and $S$ contains a non-neighbor $u$ of $v$, and Type~$4$: $S$ contains $v$,~$S \setminus \{ v \}$ is not maximal in $G'$, and $S$ is contained in the neighborhood of~$v$, that is,~$S \subseteq N_G[v]$.

First and foremost, all maximal non-induced bicliques of Type~1 and Type~2 can be enumerated from $\mathcal{S}'$ in $ \vert \mathcal{S}' \vert  \cdot n^2$ time.
We claim that there are at most~$2^{\gamma - 1} n$ maximal non-induced bicliques of Type 3:
Let $U$ be such a non-induced biclique with a bipartition $(S, T)$.
Without loss of generality, assume that $u, v \in S$.
There are at most $n$ choices for $u \in S \setminus N_G[v]{}$ and there are at most~$2^{\gamma - 1}$ choices for $T \subseteq N_G(v) \cap N_G(u)$.
Since~$U$ is a maximal non-induced biclique, we obtain $S = \bigcap_{w \in T} N_G(w)$.
Finally, there is only one maximal non-induced biclique of Type 4, namely $N_G[v]$.
Thus, we obtain the following theorem.

\begin{theorem}
  \label{thm:nbiclique-enum}
  All maximal non-induced bicliques can be enumerated in $\Oh^*(2^\gamma)$ time.
\end{theorem}

Second, we consider the decision variant of this problem.  
We show that \textsc{Non-Induced $(k_1, k_2)$-Biclique} can be solved in $\Oh^*(2^\gamma)$ time, using this enumeration algorithm.
\begin{theorem}
  \label{thm-non-induced-k-k-biclique-weakly-gamma}
  \textsc{Non-Induced~$(k_1, k_2)$-Biclique} can be solved in $\Oh^*(2^{\gamma})$~time.
\end{theorem}
\begin{proof}
With the algorithm behind Theorem~\ref{thm:nbiclique-enum} we can enumerate the vertex sets of all maximal non-induced bicliques.
This algorithm, however, only returns the vertex set, and not a bipartition of any maximal non-induced biclique.
To check whether any of these maximal non-induced bicliques has a bipartition into sets~$S$ and~$T$ such that~$\vert S\vert \ge k_1$ and~$\vert T\vert \ge k_2$, we use the following observation:
Let~$G'$ denote the complement graph of~$G$.
Any connected component of~$G'$ is either completely contained in~$S$ or completely contained in~$T$.
Now, we can use this observation to define an instance of \textsc{Subset Sum} to check whether there exists a valid bipartition.
\textsc{Subset Sum} is formally defined as follows.

  \problemdef
  {Subset Sum}
  {A set $A = \{ a_1, \dots, a_n \}$ of $n$ positive integers and $k_1 \le k_2 \in \mathds{N}$.}
  {Is there a set $B \subseteq A$ such that $k_1 \le \sum_{b \in B} b \le k_2$?}
  
  A standard dynamic programming algorithm can solve \textsc{Subset Sum} in $\Oh(n \cdot \sum_{a \in A} a)$ time.
  To solve \textsc{Non-Induced $(k_1, k_2)$-Biclique}, we construct an instance~$(A', k_1', k_2')$ of \textsc{Subset Sum} for each maximal non-induced biclique~$U$ with~$\vert U\vert\ge k_1+k_2$ returned by the algorithm of Theorem~\ref{thm:nbiclique-enum}, where~$k_1' \coloneqq  k_1$,~$k_2' \coloneqq   \vert U \vert  - k_2$, and $A' \coloneqq  \{  \vert C_i \vert  \colon i \in [\ell] \}$ for the connected components~$C_1, \dots, C_{\ell} \subseteq V(G)$ of the complement of $G[U]$.
  Observe that~$(G, k_1, k_2)$ is a Yes-instance if and only if the constructed instance of \textsc{Subset Sum} is a Yes-instance for some maximal non-induced biclique~$U$:
note that  $k_1'$ is a lower bound and~$k_2'$ is an upper bound for the size of the smaller side of any valid bipartition and any solution~$B$ of the \textsc{Subset Sum} instance corresponds to~$S$, the smaller side of the bipartition of~$U$, and~$A\setminus B$ corresponds to the other part of the bipartition.
\end{proof}

Recall that \textsc{Non-Induced Max-Edge Biclique} can be solved by solving~$\sqrt{k}$
instances of \textsc{Non-Induced~$(k_1, k_2)$-Biclique}.
Hence, we obtain the following from Theorem~\ref{thm-non-induced-k-k-biclique-weakly-gamma}.

\begin{corollary}
\textsc{Non-Induced Max-Edge Biclique} can be solved in $\Oh^*(2^{\gamma})$~time.
\end{corollary}

\subsection{Induced Biclique}

In this subsection, we study problems where one aims to find \emph{induced} maximal bicliques fulfilling certain cardinality constraints. Formally, we consider the following.

\begin{definition}
In a graph~$G=(V,E)$ two disjoint vertex sets~$S\subseteq V$ and~$T\subseteq V$ are an \emph{induced biclique} if $G[S \cup T]$ is isomorphic to a complete bipartite graph, that is, $st\in E(G)$ for each~$s\in S$ and each~$t\in T$,~$ss'\notin E(G)$ for each~$s,s'\in S$, and~$tt'\notin E(G)$ for each~$t,t'\in T$.
\end{definition}

Gaspers et al.~\cite{GKL12} provided an $\Oh^*(3^{n/3})$-time algorithm to enumerate all maximal induced bicliques.
Moreover, all maximal induced bicliques can be enumerated in $\Oh^*(3^{(\Delta+d)/3})$~time~\cite{HM20}. 
On the negative side, it is impossible to enumerate all maximal induced bicliques in time $\Oh^*(f(d + c))$ for any function $f$ because a graph may have too many maximal induced bicliques~\cite{HM20}:
Consider the graph with a single universal vertex~$u$ and~$(n - 1) / 3$ disjoint triangles.
This graph is $3$-degenerate and $2$-closed, and it has $3^{(n - 1) / 3}$ maximal induced bicliques where one part consists of $u$.

In addition to the enumeration problem, we also study the following decision problem.
 
\problemdef
{Induced~$(k_1, k_2)$-Biclique}
{A graph $G$ and $k_1, k_2 \in \mathds{N}$ such that~$k_1\le k_2$.}
{Does~$G$ contain an induced biclique with vertex sets~$S$ and~$T$ such that~$\vert S\vert \ge k_1$ and~$\vert T\vert \ge k_2$?}
%}
%Without loss of generality, we will assume that $k_1 \le k_2$.

When $k_1 = k_2$, we will refer to the problem as \textsc{Induced $(k, k)$-Biclique}.
\textsc{Induced~$(k,k)$-Biclique} is W[1]-hard~\cite{CFK+15}. 
We also consider \textsc{Induced Max-Edge Biclique} where we demand
that~$ \vert S \vert \cdot  \vert T \vert \ge k$ instead of putting constraints on the partition
sizes.
\textsc{Induced Max-Edge Biclique} is NP-hard~\cite{P03} and W[1]-hardness with respect to the solution size~$k$ can be shown by a reduction from \textsc{Independent Set} where we attach an universal vertex.
As in the non-induced case, \textsc{Induced Max-Edge Biclique} can be solved by solving~$\sqrt{k}$ instances of \textsc{Induced~$(k_1,k_2)$-Biclique}. Thus, positive results for \textsc{Induced~$(k_1,k_2)$-Biclique} transfer to~\textsc{Induced Max-Edge Biclique}.

First, we present an FPT-algorithm for \textsc{Induced~$(k, k)$-Biclique} parameterized by~$\gamma$. 

\begin{theorem}
\label{thm-induced-k-k-biclique-weakly-gamma}
  \textsc{Induced~$(k, k)$-Biclique} can be solved in $\Oh^*(\gamma^{\Oh(\gamma)})$ time.
\end{theorem}
\begin{proof}
  Since a biclique $K_{\gamma, \gamma}$ is not weakly $\gamma$-closed, $(G, k,k)$ is a No-instance if $k \ge \gamma$.
  Moreover, \textsc{Induced~$(k, k)$-Biclique} is trivially solvable in polynomial time when~$k \le 1$.
  Hence, we may assume that $2 \le k < \gamma$.
  Let~$\sigma$ be a fixed weak closure ordering of~$G$.
  Suppose that $(S, T)$ is a solution of $(G, k)$.
  Furthermore, let~$v \in S \cup T$ be the vertex of~$S\cup T$ that appears in~$\sigma$ before all other vertices of $S \cup T$.
  We assume without loss of generality that $v$ lies in~$S$.
  Note that there are at most~$n$ choices for~$v$.
  Let~$G'$ be the graph obtained by removing all vertices preceding $v$ in~$\sigma$.
Furthermore, let~$v' \in V(G') \setminus \{ v \}$  be another vertex which is contained in~$S$.
Note that there are at most~$n$ choices for~$v'$.
Next, we determine an independent set~$T \subseteq N_{G'}(v) \cap N_{G'}(v')$ of at least $k$ vertices.
  Since $ \vert N_{G'}(v) \cap N_{G'}(v') \vert  < \gamma$, there are at most $2^\gamma$ possibilities for $T$.
  Now, it remains to find an independent set $S \subseteq \bigcap_{u \in T} N_{G'}(u)$ of size at least~$k$ in $G'$.
  Recall that \textsc{Independent Set} admits a kernel with at most $\gamma k^2$ vertices by \Cref{cor:is}, and thus this can be achieved in $\Oh^*((\gamma k^2)^k)$ time.
  Since~$k < \gamma$, the overall running time is $\Oh^*(2^\gamma \gamma^{3\gamma}) = \Oh^*(\gamma^{\Oh(\gamma)})$.
\end{proof}

For $c$-closed graphs, we show that there is a single-exponential time algorithm when~$k_1 \ge 2$.
Our algorithm is based on a reduction to a variant of \textsc{Independent Set} called \textsc{Bicolored Independent Set} \cite{CK12}.

\problemdef
  {Bicolored Independent Set}
  {A graph $G$, a partition $(V_1, V_2)$ of $V(G)$, and $k_1, k_2 \in \mathds{N}$.}
  {Is there an independent set $I \subseteq V(G)$ with $ \vert I \cap V_1 \vert  = k_1$ and~$ \vert I \cap V_2 \vert  = k_2$?}

\begin{theorem}
\label{thm-induced-k-k-biclique-c-closed}
  \textsc{Induced~$(k_1, k_2)$-Biclique} with~$k_1\ge 2$ can be solved in $\Oh^*(1.611^c)$~time.
\end{theorem}
\begin{proof}
  Let~$(G=(V,E),k_1,k_2)$ be an instance of \textsc{Induced~$(k_1, k_2)$-Biclique}.
    Since~$k_1\ge 2$ any induced biclique with~$k_1$ vertices in one partite set and with~$k_2$ vertices in the other partite set contains at least on cycle on four vertices.
  For each induced cycle $(u_S, u_T, v_S, v_T)$ on four vertices in $G$ we search the largest induced biclique containing these four vertices.
  Now, we construct an instance $(G', V_1', V_2', k_1, k_2)$ of \textsc{Bicolored Independent Set}, where
\begin{itemize}
\item $V_1' \coloneqq  N_G(u_S) \cap N_G(v_S)$,
\item $V_2' \coloneqq  N_G(u_T) \cap N_G(v_T)$, and
\item $G' \coloneqq (V_1' \cup V_2', E(G[V_1']) \cup E(G[V_2']) \cup \{ v_1' v_2' \mid v_1' \in V_1', v_2' \in V_2', v_1' v_2' \notin E(G) \})$.
\end{itemize}  
  
  In other words, $G'$ is constructed from $G[V_1' \cup V_2']$ by flipping the adjacency between $V_1'$ and $V_2'$.
  By the $c$-closure of $G$, there are at most $2c - 2$ vertices in $G'$.
  Since $v_1' \in V_1'$ and $v_2' \in V_2'$ are adjacent in $G$ if and only if they are not in $G'$, there is a $(k_1, k_2)$-biclique containing $u_S, u_T, v_S, v_T$ if and only if $(G', V_1', V_2', k_1, k_2)$ is a Yes-instance.
  Since \textsc{Bicolored Independent Set} is $\Oh^*(1.2691^n)$-time solvable on $n$-vertex graphs \cite{CK12}, we obtain an $\Oh^*(1.611^c)$-time algorithm for \textsc{Induced $(k_1, k_2)$-Biclique}.
\end{proof}

By using a reduction similar to the one in the proof of \Cref{thm-induced-k-k-biclique-c-closed}, and using the algorithm of Gaspers et al.~\cite{GKL12} to enumerate all maximal induced bicliques in $\Oh^*(3^{n/3})$~time we obtain the following.

\begin{proposition}
\label{prop-enum-max-induced-bicliques-c}
All maximal induced bicliques in which each part has at least two vertices can be enumerated in $\Oh^*(3^{2c/3})$~time.
\end{proposition}

However, even 2-closed graphs may have $\Omega(3^{n/3})$ maximal induced bicliques:
Consider the aforementioned graph proposed by Hermelin and Manoussakis~\cite{HM20}, which consists of a single universal vertex $u$ and $(n - 1) / 3$ disjoint triangles.
Observe that this graph is 2-closed and has~$3^{(n - 1) / 3}$ maximal induced bicliques where one part consists of $u$.

In contrast to our positive result for $k_1 \ge 2$ presented in Theorem~\ref{thm-induced-k-k-biclique-c-closed}, we prove that \textsc{Induced~$(1,k)$-Biclique} is NP-hard even on graphs with constant~$h$-index,~$c$-closure, and weak~$\gamma$-closure. 

%bserve that there is a simple Turing reduction from an instance~$(G,k)$ of \textsc{Induced~$k$-Edge Biclique} to \textsc{Induced~$(k_1,k_2)$-Biclique}: For each two integers~$a$ and~$b$ such that~$a\cdot b=k$ create an instance~$(G,a,b)$ of  \textsc{Induced~$(k_1,k_2)$-Biclique}

\begin{theorem}
\label{thm-induced-k-edge-biclique-hard}
\textsc{Induced Max-Edge Biclique} and \textsc{Induced~$(1, k_2)$-Biclique} remain NP-hard even on graphs with~$h$-index 4,~$c$-closure 3, and weak~$\gamma$-closure 2.
\end{theorem}
\begin{proof}
We first show the NP-hardness for \textsc{Induced Max-Edge Biclique}.
We reduce from \textsc{Independent Set}, which is NP-hard even on graphs in which each vertex has degree at most 3~\cite{Gj79}.
Recall that in \textsc{Independent Set} we are given a graph~$G$ and an integer~$k$, and ask whether~$G$ contains an independent set of size at least~$k$.
We assume that~$k\ge 10$, since otherwise the instance $(G,k)$ can be solved in polynomial time.  
We construct an instance~$(G',k')$ of \textsc{Induced Max-Edge Biclique} as follows:
We begin with a copy of~$G$.  
Then, each edge~$uv\in E(G)$ is replaced by a path on four vertices~$u, u_v, v_u$, and~$v$.
Finally, we introduce a new universal vertex~$w$ (that is, $N_{G'}[w] = V(G')$) and set~$k'\coloneqq k+\vert E(G) \vert$. 
It is easy to see that~$G'$ has~$h$-index~4 (because every vertex except $w$ has degree at most 4), is~$3$-closed and weakly~$2$-closed.  
It remains to show that~$G$ contains an independent set of size~$k$ if and only if~$G'$ contains an induced biclique with at least~$k'=k+ \vert E(G) \vert $ edges. 

Suppose that~$G$ contains an independent set~$I$ of size at least~$k$.
Then, there is an independent set $I'$ of size $k +  \vert E(G) \vert $ in $G' - w$:
Since~$I$ is an independent set, for each edge~$uv\in E(G)$ we have without loss of generality that~$u\notin I$.
Let~$F\coloneqq \{u_v\mid uv\in E(G)\text{ such that~$u\notin I$}\}$ be the union of the neighbors of these vertices~$u$ not in the independent set in paths on four vertices in~$G'$.
Then,~$I'$ is the disjoint union of~$I$ and~$F$. 
Thus, the set~$I' \cup \{ w \}$ is an induced biclique with at least~$k+ \vert E(G) \vert $ edges in~$G'$. 

Conversely, suppose that~$G'$ contains a biclique $(S, T)$ with at least~$k'=k+\vert E(G)\vert$ edges. 
Since each vertex in~$G' - w$ has degree at most~3 and~$k\ge 10$, we see that vertex~$w$ is contained in $(S, T)$. 
Without loss of generality, assume that $w \in S$. 
Since~$w$ is a universal vertex, we obtain~$S=\{w\}$. 
It follows that $T$ is an independent set of size at least $k +  \vert E(G) \vert $ in $G'$.
We may assume $\vert T \cap \{ u_v, v_u \} \vert = 1$:
For each edge~$uv\in E(G)$, the set~$T$ contains at most one of $u_v$ and~$v_u$.
If neither is in $T$, then $(T \setminus \{ u \}) \cup \{ u_v \}$ is another independent set of size $k'$.
Thus, we may assume that $\vert T \cap \{ u_v,  v_u\} \vert = 1$ for every $uv \in E(G)$.
No pair of adjacent vertices $u$ and $v$ in~$G$ are part of $T$ since otherwise $T$ contains three vertices from a path $(u, u_v, v_u, v)$.
Thus,~$T \cap V(G)$ is an independent set of size $\vert T' \vert - \vert E(G) \vert \ge k$.

Finally, note that this reduction also shows NP-hardness of \textsc{Induced $(1, k_2)$-Biclique} (let $k_2 = k'$).
\end{proof}

Together with \Cref{thm-induced-k-edge-biclique-hard}, the next theorem paints a full picture of the complexity of \textsc{Induced $(k_1, k_2)$-Biclique} with respect to the weak closure number.

\begin{theorem}
  \label{thm-induced-biclique-dicho}
  For constant $k_1 \ge 2$, \textsc{Induced $(k_1, k_2)$-Biclique} on weakly $\gamma$-closed graphs is polynomial-time solvable if $\gamma \le k_1 + 1$ and NP-hard otherwise.
  Moreover, \textsc{Induced $(1, k_2)$-Biclique} on weakly $1$-closed graphs is polynomial-time solvable.
\end{theorem}
\begin{proof}
  We start with the NP-hardness.
  We adapt the reduction in the proof of \Cref{thm-induced-k-edge-biclique-hard}:
  Instead of adding a single universal vertex $w$, we add $k_1$ universal vertices (which are pairwise nonadjacent).
  Note that the graph constructed by our reduction is weakly $(k_1 + 2)$-closed (consider an ordering in which all the universal vertices appear last).

  Our polynomial-time algorithms solve \textsc{Independent Set} on weakly 1-closed graphs as a subroutine.
  We fix a weak closure ordering $\sigma$.
  Start with $I = \emptyset$.
  %We find the last vertex $v$ in the weak closure ordering.
  In a first step, we add the last vertex~$v$ of~$\sigma$ to~$I$ and then delete $N[v]$ from the graph.
  For the correctness of this step, observe that the neighborhood of $v$ is a clique. 
  Otherwise, there exists a non-neighbor~$u$ of~$v$ with~$u<_\sigma v$ and distance 2 to~$v$. 
  Since~$u$ and~$v$ have at least one common neighbor, we obtain a contradiction to the fact that the graph is weakly~$1$-closed. 
  Since~$N[v]$ is a clique, there exists a maximum independent set containing~$v$.
  We repeat this step until the graph is empty.

  Next, we give a polynomial-time algorithm for \textsc{Induced $(1, k_2)$-Biclique} on weakly 1-closed graphs.
  Without loss of generality, we assume that the input graph~$G$ is connected.
  Observe that there is a universal vertex $u$ that is adjacent to every other vertex.
  Now, observe that there is an induced $(1, k_2)$-biclique in $G$ if and only if a maximum independent set of size $k_2$ in~$G - u$.
  Since a maximum independent set in a weakly~$1$-closed graph can be found in polynomial time, we are done.

  Finally, we prove the polynomial-time solvability for $\gamma \le k_1 + 1$.
  Observe that if~$\gamma \le k_1$, then we have a No-instance of \textsc{Induced $(k_1, k_2)$-Biclique} since an induced~$(k_1,k_2)$-biclique has weak closure~$k_1+1$.
  Hence, in the following we assume that~$\gamma=k_1+1$.
  Now, consider a hypothetical solution~$(S, T)$ with~$\vert S \vert = k_1$ and~$ \vert T \vert = k_2$.
  We can guess which vertices correspond to the smaller side $S$ in $\Oh(n^{k_1})$~time.
  Let~$\sigma$ be a fixed weak closure ordering and let~$X$ be the set of vertices that occur in~$\sigma$~before any vertex in $S$.
  Since~$T\cap X$ are common neighbors of~$S$ we observe that~$T\cap X$ has size at most~$\gamma-1=k_1$.
  Hence, in $\Oh(n^{k_1})$~time, we can guess~$T\cap X$.
  It remains to find $T \cap X$.
  Note that $T \cap X \subseteq U := \bigcap_{s \in S} N(s)$, that is,~$S \subseteq N(t)$ for every vertex $t \in T \cap X$.
  Observe that $G[U \cap X]$ is weakly 1-closed:
  In the ordering $\sigma$, two nonadjacent vertices~$u,u'\in U\cap X$ such that~$u <_{\sigma} u'$ have no common neighbor $w \in U \cap X$ with $u <_{\sigma} w$ since~$u$ and~$u'$ have $S$ as common neighbors which appear after~$u'$ in~$\sigma$, and $S$ has size~$k_1=\gamma - 1$.
  As argued above, we can find a maximum independent set in~$G[U \cap X]$ in polynomial time.
  Thus, \textsc{Induced $(k_1, k_2)$-Biclique} can be solved in polynomial time if~$k_1$ is a constant and~$\gamma\le k_1+1$.
\end{proof}

To complete the dichotomy with respect to~$c$, we prove that \textsc{Induced Max-Edge Biclique} and \textsc{Induced $(k_1, k_2)$-Biclique} can be solved in polynomial time if~$c=2$. 
Observe that \Cref{thm-induced-k-k-biclique-c-closed} implies a polynomial-time algorithm for~$k_1\ge 2$ if~$c=2$.
Hence, it remains to show that \textsc{Induced $(1, k_2)$-Biclique} can be solved in polynomial-time if~$c=2$.
For this, is it sufficient to consider diamond-free graphs since each $2$-closed graph is diamond-free.

\begin{proposition}\label{prop-diamond-free}
  \textsc{Induced $(1, k_2)$-Biclique} can be solved in polynomial time on diamond-free graphs.
\end{proposition}
\begin{proof}
  Suppose that the input graph $G$ is diamond-free.
  Then, for each vertex $v \in V(G)$ the graph~$G[N(v)]$ is a disjoint union of cliques.
  Thus, $(G, 1, k_2)$ is a Yes-instance if and only if there is a vertex $v \in V(G)$ such that $G[N(v)]$ has at least $k_2$ connected components.
\end{proof}

Now, from \Cref{prop-diamond-free} ($k_1=1$) and \Cref{thm-induced-k-k-biclique-c-closed} ($k_1\ge 2$) we obtain the following.  

\begin{corollary}\label{cor:induced-biclique-2-closed}
  \textsc{Induced $(k_1, k_2)$-Biclique} and \textsc{Induced Max-Edge Biclique} can be solved in polynomial time on 2-closed graphs.
\end{corollary}

Our results for \textsc{Induced $(k_1, k_2)$-Biclique} can be summarized as follows (see also \Cref{tab:results}):
If $k_1 = k_2$, then the problem becomes FPT with respect to the weak closure number $\gamma$ (\Cref{thm-non-induced-k-k-biclique-weakly-gamma}).
In the general case, the complexity strongly depends on whether $k_1 \ge 2$ or $k_1 = 1$.
If $k_1 \ge 2$, the problem is polynomial-time solvable for $\gamma \le k_1 + 1$ (\Cref{thm-induced-biclique-dicho}), NP-hard for $\gamma \ge k_1 + 2$ (\Cref{thm-induced-biclique-dicho}), and FPT for the parameterization by $c$ (\Cref{thm-induced-k-k-biclique-c-closed}).
If $k_1 = 1$, then we have a complexity dichotomies in terms of $c$ and $\gamma$:
we have a polynomial-time algorithm for $c = 2$ (\Cref{cor:induced-biclique-2-closed}) and $\gamma = 1$ (\Cref{thm-induced-biclique-dicho}) and NP-hardness for $c \ge 3$ (\Cref{thm-induced-k-edge-biclique-hard}) and $\gamma \ge 2$ (\Cref{thm-induced-k-edge-biclique-hard}).

\section{Variants of Dominating Set}

In companion work \cite{KKS20}, we showed that \textsc{Dominating Set} admits a kernel of size $k^{\Oh(c)}$.
%We were not able to resolve the parameterized complexity of \textsc{Dominating Set} in weakly $\gamma$-closed graphs.
Recently, Lokshtanov and Surianarayanan showed that \textsc{Dominating Set} parameterized by~$\gamma+k$ can be solved in $\Oh^*(k^{\Oh(\gamma^2k^3)})$~time~\cite{LS21}.
Here, we develop FPT-algorithms for the related \textsc{Independent Dominating Set} and \textsc{Dominating Clique} problems in weakly~$\gamma$-closed graphs. 

\subsection{Independent Dominating Set}

We consider the \textsc{Independent Dominating Set} problem. The task in this problem is to find a small independent set~$S$ that dominats all vertices in~$G$.
\begin{definition}
  A vertex set~$S\subseteq V(G)$ is a dominating set in~$G$ if ~$S\cap N[v]\neq\emptyset$ for each~$v\in V$. Moreover,~$S\subseteq V(G)$ is an independent dominating set in $G$ if~$S$ is a dominating set and all vertices of~$S$ are pairwise nonadjacent. 
\end{definition}
\problemdef{Independent Dominating Set}
{A graph~$G$ and~$k \in \mathds{N}$.}
{Does~$G$ contain an independent dominating set~$S \subseteq V(G)$ of size at most~$k$?}

\textsc{Independent Dominating Set} is W[2]-hard for the parameter~$k$~\cite{DF13}. 
There are several fixed-parameter tractability results in restricted graph classes: 
% When the graph~$G$ does not contain the complete bipartite graph~$K_{i,j}$ for fixed~$j\le i$ as a (not necessarily induced) subgraph, \textsc{Independent Dominating Set} admits a kernel of~$\Oh(jk^i)$ vertices which can be computed in $\Oh(n^i)$~time~\cite{PRS12}.  
\textsc{Independent Dominating Set} has a kernel of~$\Oh(d^2 k^{d+1})$~vertices computable in $\Oh^*(2^d)$~time~\cite{PRS12}.
% This kernel is essentially tight since \textsc{Independent Dominating Set} in~$d$-degenerate graphs admits no kernel of size~$\Oh(k^{d-4-\epsilon})$ for any~$\epsilon>0$ unless NP~$\subseteq$ coNP/poly~\cite{CGH17}. 
Moreover, when the graph contains no cycles of length 3 or 4, \textsc{Independent Dominating Set} can be solved in $\Oh^*(k^{\Oh(k)})$~time~\cite{RS08}.

We present an FPT-algorithm \textit{SolveIDS} (\Cref{algo-branching-ids}) with running time~$\Oh^*((\frac{\gamma-1}{2})^k k^{2k})$.
Note that our algorithm extends the $\Oh^*(k^{\Oh(k)})$~time algorithm of Raman and Saurabh \cite{RS08}, because any graph without cycles of length~3 or~4 is 2-closed.
Let~$G'$ be a copy of~$G$.
\Cref{algo-branching-ids} first greedily computes an independent set $I$ of~$G$ of size at most~$k+1$ by iteratively choosing vertices $v$ such that $\cl_{G'}(v) \le \gamma - 1$ and afterwards removing~$N[v]$ from~$G'$ (Line~\ref{line-ids-choose-vertex}).
If~$I$ is inclusion-maximal and of size at most~$k$, then~$I$ constitutes a solution.
Otherwise, we find a vertex set~$P$ to branch on.
The choice of~$I$ will ensure that $P$ has at most $(\gamma - 1) \binom{k + 1}{2}$ vertices.

\begin{algorithm}[t]
	\caption{An FPT-Algorithm for \textsc{Independent Dominating Set}.}
	\label{algo-branching-ids}
	\begin{algorithmic}[1]	
	\Procedure{\textit{SolveIDS}}{$G,k$}
	\State \algorithmicif\ $k=0$ \ALGAND $V(G)\neq \emptyset$ \algorithmicthen\ \Return No  \label{line-ids-no}
	\State Let $I\coloneqq \emptyset$ and $G'\coloneqq G$.
  \LineComment{$I$ will be an independent set of size at most~$k+1$ in~$G $}\label{line-compute-is-start}
	\While {$V(G')\neq\emptyset$ \ALGAND $ \vert I \vert \le k$}
	\State Let~$v$ be a vertex such that $\cl_{G'}(v) \le \gamma - 1$. \label{line-ids-choose-vertex}
	\State $I\coloneqq I\cup\{v\}$ and $G'\coloneqq G' - N_{G'}[v]$. \label{line-compute-is-end}
	\EndWhile 
	\If {$ \vert I \vert  \le k$} \Return Yes \label{line-brute-force}
	\Else \label{line-I-small}
    \State $P\coloneqq \{v~ \vert ~v \text{ is a common neighbor of at least two vertices in } I\}$ \label{line-compute-common-neigs-of-i}
    \ForEach{$u \in P$}
      \If {\textit{SolveIDS}$(G- N_G[u],k-1)$ returns Yes}
        \State \Return Yes \label{line-choose-v-in-p}
      \EndIf
    \EndFor
	\EndIf
	\State \Return No
	\EndProcedure
	\end{algorithmic}
\end{algorithm}

% The input of \textit{SolveIDS} is a weakly~$\gamma$-closed graph~$H=(V,E)$, and a parameter~$k'\le k$. 
% By~$\sigma$ we denote a fixed weakly~$\gamma$-closure ordering of~$H$. 
% The initial call is \textit{SolveIDS}$(G,k)$. 
% %The goal is to find an independent dominating set~$D'$ such that~$D \subseteq D'$ and~$ \vert D' \vert  \le  \vert D \vert  + k'$. 
% First, in Line~\ref{line-ids-no} we check whether~$H$ is not empty and has parameter~$k'=0$. 
% Second, in Lines~\ref{line-compute-is-start}--\ref{line-compute-is-end}, an independent set~$I$ of~$H$ is computed as follows: 
% While~$I$ has not size~$k'+1$, the vertex~$v\in V(H)\setminus N[I]$ with smallest index in the weakly~$\gamma$-closure ordering~$\sigma$ in~$H-N[I]$ is added to~$I$. 
% %Initially,~$I$ is empty. 
% %While~$W\coloneqq N[I]\neq V$ and~$ \vert I \vert \le k'$, the vertex with lowest index in the weakly~$\gamma$-%ordering of~$G[W]$ is added to the independent set~$I$. 

% If~$ \vert I \vert =k'+1$, \textit{SolveIDS} computes the set~$P$ of vertices with at least two neighbors in~$I$. 
% Now, for each vertex~$u\in P$, \textit{SolveIDS} branches in Line~\ref{line-choose-v-in-p} to add~$u$ to the independent dominating set. 
% Otherwise,~$ \vert I \vert \le k'$. 
% Then,~$(H,k')$ is a Yes-instance. 

\begin{theorem}
\label{theo-ids-fpt-algo}
\textsc{Independent Dominating Set} can be solved in $\Oh^*((\frac{\gamma-1}{2})^k k^{2k})$~time. 
\end{theorem}

\begin{proof}
We show that the search tree algorithm \Cref{algo-branching-ids} solves any instance $(G, k)$ of \textsc{Independent Dominating Set} in the claimed time.
First, we prove the correctness of \Cref{algo-branching-ids}.
Let~$I$ be the independent set of size at most~$k+1$ of~$G$ obtained in Lines~\ref{line-compute-is-start} to \ref{line-compute-is-end}.
Suppose that~$ \vert I \vert \le k$.  
Since~$I$ is a maximal independent set, each vertex~$v\in V(G)$ is either contained in~$I$ or a neighbor of a vertex in~$I$. 
Hence,~$I$ is an independent dominating set of size at most~$k$ of~$G$. 
Thus,~$(G,k)$ is a Yes-instance.
Now, suppose that~$ \vert I \vert =k+1$. 
Let~$P$ be the set of vertices in~$G$ which have at least two neighbors in~$I$ (Line~\ref{line-compute-common-neigs-of-i}). 
Since~$ \vert I \vert =k+1$, the sought solution $S$ must contain at least one vertex $u$ of~$P$. 
If $u \in S$, then $S$ does not contain any neighbor of $u$.
Thus, the branching into $(G - N_G[u], k - 1)$ in Line~\ref{line-choose-v-in-p} is correct. 
%Since~$H$ does not contain any neighbor of~$D$, we observe that~$D\cup I$ is an independent dominating set of size at most~$k$ of~$G$. 

Now, we analyze the running time of \Cref{algo-branching-ids}.
First, we bound the number of children of any node in the search tree.
To do so, we prove that~$ \vert P \vert \le (\gamma-1) \binom{k + 1}{2}$. % where $P$ is the set of vertices in~$G$ which have at least two neighbors in~$I$ (Line~\ref{line-compute-common-neigs-of-i}).
Let~$v_i$ be the $i$th vertex added to $I$ in Line~\ref{line-compute-is-end} and let $G_i \coloneqq  G - N_G[\{ v_1, \dots, v_{i - 1} \}]$ for each $i \in [k + 1]$.
Observe that~$\cl_{G_i}(v_i) \le \gamma - 1$ for each $i \in [k]$ since~$G$ is weakly~$\gamma$-closed.
For a vertex $u \in P$, let $v_i \in I$ be the first vertex that $u$ is adjacent to.
Then,~$u$ is present in the graph $G_i$ and we have $u \in N_{G_i}(v_i)$.
Thus, $P \subseteq \bigcup_{i \in [k+1]} N_{G_i}(v_i)$ and we see that $ \vert P \vert  \le \sum_{i \in [k + 1]}  \vert N_{G_i}(v_i) \cap P \vert $.
Moreover, we have $ \vert N_{G_i}(v_i) \cap P \vert  \le \sum_{j \in [i + 1, k + 1]}  \vert N_{G_i}(v_i) \cap N_{G_i}(v_j) \vert $ for each $i \in [k]$.
Since~$I$ is an independent set we obtain that~$\vert N_{G_i}(v_i) \cap N_{G_i}(v_j) \vert\le \gamma-1$ for each~$j\in[i+1,k+1]$.
Therefore, 
\begin{align*}
   \vert P \vert  \le \sum_{i < j \in [k + 1]}  \vert N_{G_i}(v_i) \cap N_{G_i}(v_j) \vert  \le (\gamma - 1) \binom{k + 1}{2}.
\end{align*}
%Here, the last inequality is due to the fact that $\cl_{G_i}(v_i) \le \gamma - 1$ for each $i \in [k]$.

It is easy to see that finding an independent set~$I$ in Lines~\ref{line-compute-is-start} to \ref{line-compute-is-end} only requires polynomial time. 
Hence, the algorithm spends polynomial time in each search tree node.
Since each node has at most~$(\gamma-1) \binom{k + 1}{2}$ children in the search tree and its depth is at most $k$, the overall running time of \Cref{algo-branching-ids} is $\Oh^*((\gamma-1 \cdot \binom{k + 1}{2})^k) = \Oh^*((\frac{\gamma-1}{2})^k k^{2k})$. 
\end{proof}

A natural next question is whether \textsc{Independent Dominating Set} admits a polynomial kernel in weakly~$\gamma$-closed graphs.
We answer this question in the negative way, that is, we provide kernel lower bounds for \textsc{Independent Dominating Set} via a \emph{cross-composition} \cite{BDFH09,BJK14}.

An equivalence relation~$R$ on~$\Sigma^*$ is called a \emph{polynomial equivalence relation} if the following two conditions hold:
(1) There is an algorithm that given two strings~$x,y\in\Sigma^*$ decides whether~$x$ and~$y$ belong to the same equivalence class in $( \vert x \vert + \vert y \vert )^{\Oh(1)}$~time, and
(2) for any finite set~$S\subseteq \Sigma^*$ the equivalence relation~$R$ partitions the elements of~$S$ into at most~$(\max_{x\in S} \vert x \vert )^{\Oh(1)}$ classes.

\begin{definition}
\label{def:cross-comp}
Let~$L\subseteq\Sigma^*$ be a set and let~$Q\subseteq\Sigma^*\times\mathds{N}$ be a parameterized problem. 
We say that~$L$ \emph{cross-composes} into~$Q$ if there is a polynomial equivalence relation~$R$ and an algorithm which, given~$2^t$ strings~$x_1,x_2,\ldots ,x_{2^t}$ belonging to the same equivalence class of~$R$, computes an instance~$(x^*,k^*)\in\Sigma^*\times\mathds{N}$ in time polynomial in~$\sum_{i=1}^{2^t} \vert x_i \vert $ such that:
\begin{enumerate}
\item $(x^*,k^*)\in Q$ if and only if~$x_i\in L$ for some~$i\in[2^t]$, and
\item $k^*$ is bounded by a polynomial in~$\max_{i \in [2^t]}  \vert x_i \vert +t$.
\end{enumerate}
\end{definition}

It is known that if an NP-hard problem cross-composes into a parameterized problem, then the parameterized problem does not admit a kernel of polynomial size unless \PHC{} \cite{BDFH09,BJK14}.

\begin{theorem}
\label{thm-ids-no-poly-gamma-2}
Unless \PHC{}, \textsc{Independent Dominating Set} admits
\begin{itemize}
  \item no kernel of size $(k + c)^{\Oh(1)}$ and
  \item no kernel of size $k^{\Oh(1)}$ even if $\gamma = 2$.
\end{itemize}

\end{theorem}

\begin{proof}
We provide a cross-composition from \textsc{Independent Dominating Set} on $2$-closed graphs.
Note that \textsc{Independent Dominating Set} remains NP-hard on 2-closed graphs and hence also on weakly~$2$-closed graphs.
This follows from the fact that \textsc{Independent Dominating Set} is NP-hard on graphs of girth at least five~\cite{CP17,ZZ95}.
In particular, the graph constructed by the cross-composition procedure is weakly 2-closed and $(t + 2)$-closed.

Assume that we are given~$2^t$ instances~$I_x\coloneqq (G_x,k)$ of \textsc{Independent Dominating Set} on 2-closed graphs for $x\in [2^t]$.
We will describe how to construct an instance~$(G',k')$ of \textsc{Independent Dominating Set} with weak closure~$2$ and closure~$t+2$ that meets the requirements as specified in \Cref{def:cross-comp}.
To do so, we write an integer $x \in [2^t]$ in binary encoding $(x_1,\ldots, x_t)$.
For~$y\in\{0,1\}$ let~$\widetilde{y}\coloneqq  1 - y$.
Furthermore, for any string~$x\coloneqq (x_1,\ldots , x_s)$ a string~$(x_1,\ldots, x_p)$ for some~$p\in[s]$ is a \emph{prefix} of~$x$.

First, we construct the \emph{instance selector} gadget $H_t$.
For each string~$z\in\{0,1\}^*$ of length at most~$t$, we introduce a vertex~$w_z$ to~$H_t$.
We add an edge~$w_{z} w_{y}$ whenever~$z$ is a prefix of~$y$.
Furthermore, we add an edge~$w_{z}w_{y}$ whenever $z$ and $y$ are of the same length $s$ and they differ only in the last bit, that is, $z\coloneqq (z_1,\ldots,z_{s-1},z_s)$ and~$y\coloneqq (z_1,\ldots, z_{s-1},\widetilde{z_s})$.
This concludes the construction of $H_t$.

To construct~$G'$, start with a disjoint union of~$G_x$ for all~$x \in [2^t]$ and with~$H_t$.
We then add an edge from vertex~$w_z$ to every vertex in~$V(G_x)$ whenever~$z$ is a prefix of~$x$.
Finally, we set~$k'\coloneqq k+t$.
As we will show, every independent set of $H_t$ avoids dominating the vertices of $G_x$ for some $x \in \{ 0, 1 \}^t$ (see \Cref{claim:select}).
Intuitively speaking, this ensures that $G_x$ has an independent dominating set of size at most $k$ whenever $(G', k')$ is a Yes-instance.

Before showing the correctness, we verify that~$G'$ is weakly~$2$-closed.
To this end, we show that every induced subgraph~$G^*$ of~$G'$ has a vertex~$v\in V(G^*)$ such that~$\cl_{G^*}(v)<2$.
If~$G^*$ does not contain any vertex of $H_t$, then~$G^*$ is $2$-closed.
Otherwise, assume that~$G^*$ contains at least one vertex of $H_t$.
Let~$w_z\in V(G^*)$ be a vertex such that~$z\coloneqq (z_1,\ldots, z_{s-1},z_s)$ has the shortest length among all vertices~$w_y\in V(H_t)\cap V(G^*)$.
We show that~$w_z$ and any vertex~$v \in V(G^*) \setminus N[w_z]$ have at most one common neighbor.
We show this claim for~$v=w_y$; the proof for~$v \in V(G_x)$ with $x \in [2^t]$ is analogous because $N(v) \subseteq N(w_x)$.
Since~$w_zw_y\notin E(G')$, we observe that~$z$ is not a prefix of~$y$.
By construction, we have 
\begin{align*}
N_{G'}(w_z)\cap N_{G'}(w_y) \subseteq Z \cup \{ w_{(z_1, \dots, \widetilde{z_s})}\},
\end{align*}
Here,~$Z \coloneqq  \{ w_{z'} \mid z' \text{ is a prefix of both $y$ and $z$} \}$.
Since~$z$ has the shortest length among all strings~$z'$ such that~$w_{z'}\in V(G^*)$, we have $Z \cap V(G^*) = \emptyset$.
Hence, we have $ \vert N_{G^*}(w_z)\cap N_{G^*}(w_y) \vert  \le 1$.
We thus have shown that~$G'$ is weakly~$2$-closed.

We then examine the $c$-closure of $G'$.
By construction, each vertex has at most~$t$ neighbors in~$H_t$.
Hence, any two nonadjacent vertices of $G'$ have at most $t$ common neighbors in $H_t$.
Moreover, since the~$2^t$ many instances of \textsc{Independent Dominating Set} are~$2$-closed and disjoint and any two vertices in~$H_t$ having a common neighbor in some~$G_x$ are adjacent, we conclude that any two nonadjacent vertices of~$G'$ have at most one common neighbor in $\bigcup_{x \in [2^t]} V(G_x)$.
Thus, $G'$ is $(t + 2)$-closed.

Next we show that~$G'$ contains an independent dominating set of size at most~$k'$ if and only if~$G_x$ contains an independent dominating set of size at most~$k$ for some~$x\in[2^t]$.

Assume that~$S$ is an independent dominating set of size at most~$k$ for the instance~$I_x$.
Recall that the binary encoding of~$x$ is~$(x_1,\ldots, x_t)$ where~$x_i\in\{0,1\}$ for each~$i\in[t]$.
We define~$y_i\coloneqq (x_1,\ldots, x_{i-1},\widetilde{x_i})$ for each~$i\in[t]$.
In the following, we verify that~$S'\coloneqq S\cup\{w_{y_i}\mid i\in[t]\}$ is an independent dominating set of size at most~$k'$ of~$G'$.
Clearly,~$S'$ has size at most~$k+t=k'$.
It remains to verify that~$S'$ is an independent dominating set of~$G'$.

First, we show that~$S'$ is an independent set in~$G'$.
Since~$y_i$ is not a prefix of~$x$ for each~$i\in[t]$, we conclude that~$z_{y_i}$ is not adjacent to any vertex of~$V(G_x)$.
Furthermore, by assumption~$S$ is an independent set.
Thus, it remains to show that~$w_{y_i}w_{y_j}\notin E(G')$ for each~$i\ne j$.
Without loss of generality, assume that~$i<j$.
Recall that~$y_i=(x_1,\ldots,x_{i-1},\widetilde{x_i})$ and that~$y_j=(x_1,\ldots, x_{i-1},x_i,\ldots, x_{j-1},\widetilde{x_j})$.
Hence,~$y_i$ is not a prefix of~$y_j$.
Thus,~$S'$ is an independent set.

Second, we show that~$S'$ is a dominating set of~$G'$.
First, we show that~$S'$ dominates $w_z$ for every vertex of $H_t$.
Assume that~$z$ is a prefix of~$x$.
Let~$y\in S$.
Since~$S\subseteq V(G_x)$,~$S\subseteq S'$, and since~$w_zq\in E(G')$ we see that~$y$ dominates~$w_z$.
Otherwise, assume that~$z$ is no prefix of~$x$.
Let $i \in [t]$ be the smallest number such that~$x_i\ne z_i$.
Then,~$z=(x_1,\ldots,x_{i-1},\widetilde{x_i},z_{i+1},\ldots z_s)$ where~$s$ denotes the length of~$z$.
Observe that~$y_i=(x_1,\ldots, x_{i-1},\widetilde{x_i})$ is a prefix of~$z$.
Hence, $w_z$ is dominated by~$w_{y_i}\in S'$.
We can analogously show that every vertex in~$v\in V(G_y)$ for some~$y \in [2^t]$ is dominated by $S'$: either~$y=x$, then~$v$ is dominated by some vertex in~$S\subseteq S'$ since~$S$ is an independent dominating set of~$G_x$, or~$y\ne x$, then there exists a prefix~$z$ of~$y$ such that~$w_z\in S'$ and~$vw_z\in E(G')$.
Thus, we have shown that $S'$ is a dominating set of~$G'$.

Conversely, suppose that~$G'$ has an independent dominating set~$S'$ of size at most~$k'=k+t$.
We prove that $G_x$ has an independent dominating set of size $k$ for some $x \in [2^t]$.
We start with an observation on independent sets in the instance selector gadget~$H_t$.

\begin{claim}
  \label{claim:select}
  Let $I'$ be an independent set of~$H_t$.
  Then, there exists an~$x=(x_1,\ldots, x_t)$ such that for any prefix~$y$ of~$x$, the vertex~$w_y$ is not contained in~$I'$.
\end{claim}
\begin{claimproof}
% For the other direction, 
% We first show that there exists an~$x=(x_1,\ldots, x_t)$ such that the vertex $w_y$ is not contained in~$I'$ for any prefix~$y$ of~$x$.
We construct the string~$x=(x_1, \ldots, x_t)$ inductively.
First, we construct~$x_1$ for the start of the induction.
Observe that~$w_{(0)}$ and~$w_{(1)}$ are adjacent in~$G'$ (note that these correspond to the 1-bit strings).
Hence,~$I'$ can contain at most one of these two vertices.
In other words,~$w_i\notin I'$ for some~$i\in\{0,1\}$.
We set~$x_1\coloneqq i$.
Now, we consider the inductive step.
Here we assume that we already constructed the string~$(x_1,\ldots, x_s)$ for some~$s\in[t-1]$ and now we aim to construct~$x_{s+1}$.
Observe that for~$y\coloneqq (x_1, \ldots, x_s, 0)$ and~$z\coloneqq (x_1, \ldots, x_s,1)$ the vertices~$w_y$ and~$w_z$ are adjacent in~$G'$.
Hence,~$I'$ can contain at most one of these two vertices.
In other words,~$w_{(x_1, \ldots, x_s,i)}\notin I'$ for some~$i\in \{0,1\}$.
We set~$x_{s+1}\coloneqq i$.
Now, the claim follows after constructing~$x_t$.
\end{claimproof}

Let~$I'\coloneqq S'\cap V(H_t)$.
In the following, let~$x=(x_1,\ldots, x_t)$ be a string fulfilling the conditions of Claim~\ref{claim:select}, that is, for any prefix~$y$ of~$x$, the vertex~$w_y$ is not contained in~$I'$ and hence also not in~$S'$.
Furthermore, let~$G_x$ be the graph of the \textsc{Independent Dominating Set} instance corresponding to~$x$.
Since~$w_y\notin S'$ for any prefix~$y$ of~$x$, we obtain that~$S'\cap V(G_x)\ne\emptyset$.
% Since $w_z$ is adjacent to every vertex of $G_x$ for any prefix $z$ of $x$, we have~$w_z \notin S'$.
Consider the vertex $w_{y_i}$ where $y_i \coloneqq (x_1,\ldots, x_{i-1},\widetilde{x_i})$ for some~$i\in[t]$.
By construction, the vertices~$\{w_{y_i}\mid i\in[t]\}$ are pairwise unreachable in $G' - N[S' \cap V(G_x)]$.
Hence,~$G' - N[S' \cap V(G_x)]$ has at least $t$ connected components.
Since $S'$ contains at least one vertex of each connected component, it follows that there is an independent dominating set in $G_x$ of size at most $k' - t = k$.
\end{proof}

%In previous work \cite{KKS20}, we have shown that \textsc{Dominating Set} has a kernel of size $k^{\Oh(c)}$---polynomial in $k$ when $c$ is constant.
%We leave it open whether \textsc{Independent Dominating Set} admits a polynomial kernel for constant $c$.

\subsection{An FPT-Algorithm for Dominating Clique}
\label{sec:dc}

We now consider the \textsc{Dominating Clique} problem.
The task in this problem is to find a small clique that dominates all vertices.
\begin{definition}
  A set~$S\subseteq V(G)$ is a dominating clique in~$G$ if all vertices of~$S$ are pairwaise adjacent and~$S$ is a dominating set.
\end{definition}

\problemdef{Dominating Clique}
{A graph~$G$ and a parameter~$k \in \mathds{N}$.}
{Does~$G$ contain a dominating clique of size at most~$k$?}

It is known that \textsc{Dominating Clique} is W[2]-hard with respect to~$k$ even on graphs which do not contain a~$4$-claw (a~$K_{1,4}$) as an induced subgraph~\cite{CPPPW11}.

Note that there is a straightforward $\Oh^*(d^k)$-time algorithm for \textsc{Dominating Clique} on~$d$-degenerate graphs:
Enumerate all cliques of size at most~$k$ and check if any of them dominates all vertices. To see the running time bound, observe that we may use a degeneracy ordering~$(v_1, \dots, v_n)$ of~$G$ and recall that for this ordering $\deg_{G_i}(v_i) \le d$ where~$G_i\coloneqq G[\{v_i,\ldots, v_n\}]$.
By considering all~$n$ possibilities for the first vertex of the dominating clique in this ordering, we can enumerate every clique of at most $k$ vertices in $\Oh(n\cdot d^k)$ time.
Instead, one may also solve \textsc{Dominating Clique} in $\Oh^*(2^d)$~time by enumerating all cliques of~$G$. At first glance, an $\Oh^*(2^d)$-time algorithm may sometimes seem prefrable to the $\Oh^*(d^k)$-time algorithm. However, a more precise running time bound of the latter algorithm is~$\Oh^*(\binom{d}{k})$ which is never larger than~$\Oh^*(2^d)$.

In this subsection, we describe an FPT-algorithm for weakly $\gamma$-closed graphs, resulting in an $\Oh^*((\gamma - 1)^{k})$-time algorithm.
Note that a maximal clique of a weakly $\gamma$-closed graph may be arbitrarily large.
Thus, a simple brute-force search on maximal cliques may require $\Omega(n^k)$ time even on graphs with constant weak closure. 
Moreover, we want to avoid enumerating all maximal cliques since this alone incurs a running time of~$\Omega(3^{\gamma/3})$~\cite{FRSWW20}.
Instead, we will use \Cref{algo-branching-dc} for each vertex~$v_i$ in a fixed weak closure ordering~$\sigma$. 
The key idea is that we assume that~$v_i$ is the first vertex in the dominating clique with respect to~$\sigma$. 
As we shall see in the proof of \Cref{thm-dom-clique-weakly-gamma-fpt} this guarantees that for each vertex~$w$ which is not adjacent to~$v_i$, we may branch into at most~$\gamma-1$ cases to determine a vertex that dominates~$w$.

\begin{algorithm}[t]
  \caption{An algorithm for finding a dominating clique~$S$ that contains $v_i$ as the first vertex in the fixed weak closure ordering~$\sigma$ of~$G$. Initially we have~$T\coloneqq \{v_i\}$.}
  \label{algo-branching-dc}
  \begin{algorithmic}[1]  
  \Procedure{\textit{SolveDC}}{$G, k, T$} \Comment{$T\subseteq\{v_i, \ldots , v_n\}$ and~$v_i\in T$}
  \State \algorithmicif\ $k=0$ \ALGAND~$V(G)\neq N[T]$ \algorithmicthen\ \Return No \label{line-cdc-return-no}
  \State \algorithmicif\ $V(G) = N[T]$ \algorithmicthen\ \Return Yes\label{line-cdc-yes}
  \State Find a vertex~$w$ such that~$v_iw\notin E(G)$ \label{line-cdc-compute-v-j}
  \ForEach{$u\in \bigcap_{x\in T}N(x)\cap N(w)\cap V(G_i)$} \label{line-cdc-branch1} \Comment {$G_i\coloneqq G[\{v_i, \ldots , v_n\}]$}
    \If {\textit{SolveDC}$(G, k - 1, T\cup\{u\})$ returns Yes}
      \State \Return Yes \label{line-cdc-branch2}
    \EndIf
  \EndFor
  \State \Return No
  \EndProcedure
  \end{algorithmic}
\end{algorithm}

\begin{theorem}
\label{thm-dom-clique-weakly-gamma-fpt}
\textsc{Dominating Clique} can be solved in $\Oh^*((\gamma -1)^{k})$~time.
\end{theorem}

\begin{proof}

To solve an instance $(G, k)$ of \textsc{Dominating Clique}, we first compute a weak closure ordering~$\sigma$. 
Afterwards, we invoke \textit{SolveDC} on input $(G, k-1, \{v_i\})$ for each vertex~$v_i\in V$. 
In the call \textit{SolveDC}$(G, k-1, \{v_i\})$, we assume that~$v_i$ is the first vertex in the dominating clique~$S$ with respect to the weak closure ordering~$\sigma$. 

We first show that \textit{SolveDC}$(G, k, T)$ is correct in the following sense: it returns Yes if and only if there is a dominating clique~$S$ of size at most $k$ which contains all vertices of~$T$, and vertex~$v_i$ is the first vertex in~$S$ with respect to~$\sigma$ (where~$v_i$ is the minimal vertex of~$T$ with respect to~$\sigma$).
It is easy to see that the terminal condition in Line~\ref{line-cdc-return-no} is correct. Moreover, Line \ref{line-cdc-yes} (where we return Yes when $V(G) = N[T]$) is correct if $T$ is a clique, we will argue below that this is always the case. 
Let~$w \notin N(v_i)$ be the vertex computed in Line~\ref{line-cdc-compute-v-j}. 
Since we want to compute a dominating clique~$S$ which contains~$T$, where vertex~$v_i$ is the first vertex in~$S$ with respect to the weak closure ordering~$\sigma$ and since~$v_iw\notin E(G)$, any dominating set must contain at least one vertex~$u$ of $N(w)\cap V(G_i)$. Moreover, since we are searching for a dominating clique, we have that~$u$ must also be a common neighbor of all vertices in~$T$, that is,~$u\in (\bigcap_{x\in T}N(x))$. 
Thus, the branching into $(G, k - 1, T \cup \{ u \})$ in Lines~\ref{line-cdc-branch1} and \ref{line-cdc-branch2} is correct. Since each vertex~$u$ chosen in Line~\ref{line-cdc-branch1} is a common neighbor of all vertices in~$T$, we conclude that~$T$ is a clique and thus Line~\ref{line-cdc-yes} returns Yes if and only if~$G$ contains a dominating clique of size at most~$k$. 
Furthermore, each vertex~$u$ chosen in Line~\ref{line-cdc-branch1} is contained in~$G_i$. 
Hence~$v_i<_\sigma u$. 
In other words, vertex~$v_i$ is the smallest vertex in~$T$ with respect to~$\sigma$.

Let us analyze the time complexity of \textit{SolveDC}. 
It is easy to see that Lines~\ref{line-cdc-return-no} to~\ref{line-cdc-compute-v-j} can be performed in polynomial time.
Consider the search tree where each node corresponds to an invocation of \textit{SolveDC}. 
We show that each node in the search tree has at most~$\gamma-1$ children. To this end, we bound the size of~$ \vert N(v_i)\cap N(w)\cap V(G_i) \vert $ which is an upper bound on the number of branches created in Line~\ref{line-cdc-branch1}. If~$v_i<_\sigma w$, then~$ \vert N(v_i)\cap N(w)\cap V(G_i) \vert \le\gamma -1$ by Definition~\ref{def:gamma}. Otherwise, if~$w<_\sigma v_i$, then~$v_i$ and~$w$ have at most~$\gamma-1$ common neighbors in~$\{v'\mid w<_\sigma v'\}$ and thus also in~$V(G_i)$.
Hence, each node has at most~$\gamma -1$ children.  
Moreover, the depth of the search tree is at most~$k-1$. 
Thus, we spend $\Oh^*((\gamma - 1)^{k-1})$ time for each vertex~$v_i\in V(G)$ and the claimed running time bound follows.
\end{proof}

In companion work~\cite{KKS20}, we showed by a reduction from \textsc{$\lambda$-Hitting Set} that  \textsc{Dominating Set} does not admit kernels of size~$\Oh(k^{c-1-\epsilon})$ under some standard complexity-theoretic assumptions.
The idea of this well-known reduction is to construct a split graph in which the universe is the clique and the sets of the set family are the vertices in the independent set. In other words, the split graph is obtained from the incidence graph of the set family by making the universe a clique. We thus directly obtain the following hardness results for \textsc{Dominating Clique} from this reduction.

\begin{proposition}
\label{thm-dom-clique-kernel-lb}
For~$c\ge 3$, \textsc{Dominating Clique} has no kernel of size~$\Oh(k^{c-1-\epsilon})$ unless coNP~$\subseteq$ NP/poly. 
\end{proposition}

\begin{proposition}
\label{thm-dom-clique-gamma-to-k-mandatory}
Unless the ETH fails, there is no $n^{o(k)}$-time algorithm for \textsc{Dominating Clique}. 
\end{proposition}

In view of \Cref{thm-dom-clique-gamma-to-k-mandatory}, it is unlikely that the running time $\Oh^*((\gamma - 1)^k)$ of Theorem~\ref{thm-dom-clique-weakly-gamma-fpt} can be substantially improved:
an algorithm running in time $\Oh^*(\gamma^{o(k)})$ or $\Oh^*(c^{o(k)})$ would dispute the ETH.
Furthermore, for~$\lambda=2$, when \textsc{$\lambda$-Hitting Set} is the \textsc{Vertex Cover} problem, the reduction shows NP-hardness for constant closure since the independent set vertices in the constructed instance have degree~2.
\begin{proposition}
\label{thm-dom-clique-nph-constant-c}
\textsc{Dominating Clique} remains NP-hard even on~$3$-closed graphs. 
\end{proposition}

Thus, both parameters~$\gamma$ and~$k$ are necessary in \Cref{thm-dom-clique-weakly-gamma-fpt}: \textsc{Dominating Clique} is W[2]-hard with respect to~$k$ even on graphs which do not contain a~$4$-claw as an induced subgraph~\cite{CPPPW11} and NP-hard even for~$\gamma=3$ since~$\gamma\le c$ (\Cref{thm-dom-clique-nph-constant-c}).

\section{Conclusion}
We have provided further applications of the weak closure parameter~$\gamma$ which was
introduced for clique enumeration~\cite{FRSWW20}.  Given the algorithmic usefulness of the class of
weakly closed graphs, it seems important to further study its properties. For example, it
would be nice to obtain a forbidden subgraph characterization. We note that the
weakly-1-closed graphs are exactly the graphs that do not contain a~$C_4$ or a~$P_4$ as an
induced subgraph. These graphs are also known as quasi-threshold graphs. Can we obtain a
similar characterization for weakly~$2$-closed graphs? 

Further FPT-algorithms for the
parameter~$\gamma$ would also be very interesting from a theoretical and practical point of view. 
In particular, obtaining kernelization algorithms for the class of weakly closed graphs is unexplored for many problems. For example, \textsc{Dominating Set} has an FPT-algorithm for
the parameter~$\gamma+k$~\cite{LS21} but it remains open whether \textsc{Dominating Set} admits a polynomial kernel for~$k$ if~$\gamma$ is a constant.
Only for special graph classes, kernels for~$\gamma+k$ are known:
In companion work~\cite{KKS21} we provided almost tight kernels of size~$k^{\Oh(\gamma)}$ for split graphs and of size~$k^{\Oh(\gamma^2)}$ for graphs with constant clique size.
In contrast, for the larger parameters degeneracy and closure almost tight upper and lower bounds are known:
\textsc{Dominating Set} admits a kernel with $\Oh(k^{(d+1)^2})$~vertices~\cite{PRS12}, and a kernel of size~$\Oh(k^{(d-3)(d-1)-\varepsilon})$ is unlikely~\cite{CGH17}.
Similarly, \textsc{Dominating Set} admits a kernel with~$k^{\Oh(c)}$ vertices~\cite{KKS20} and a kernel of size~$\Oh(k^{c-1-\varepsilon})$ is unlikely~\cite{KKS20}. Observe in this context that, for \textsc{Independent Dominating Set}, we showed that a kernel of size~$k^{g(\gamma)}$ is unlikely. Such a kernel may, however, still be achievable for \textsc{Dominating Clique}.

Also, some questions about clique relaxations in (weakly) closed graphs remain open:
In \textsc{$s$-Club} we ask for a vertex set~$S$ of size at least~$k$ in a graph~$G$ which is an $s$-club. 
Recall that in an $s$-club~$S$ each pair of vertices in~$S$ has distance at most~$s$ in~$G[S]$ (see also \Cref{def-s-club}).
We showed that \textsc{2-Club} is NP-hard even in $4$-closed graphs (\Cref{thm-2-club-nphard-constant-closure}).
It is open, whether \textsc{2-Club} is also NP-hard 2-closed or 3-closed graphs.
Finally, the complexity of \textsc{$s$-Club} for~$s\ge 3$ on graphs with constant closure remains open.

% \begin{itemize}
% \item Forbidden subgraphs of . Note that these graphs are exactly the quasi-threshold graphs.
% \item Forbidden subgraphs for weakly~$2$-closed graphs:~$\overline{\text{A}}, \overline{\text{H}}, \overline{\text{net}}, C_4$. Are these all forbidden subgraphs?
% \end{itemize}

\bibliographystyle{plain}

\begin{thebibliography}{10}

\bibitem{AG09}
Noga Alon and Shai Gutner.
\newblock Linear time algorithms for finding a dominating set of fixed size in
  degenerated graphs.
\newblock {\em Algorithmica}, 54(4):544--556, 2009.

\bibitem{HR20}
Balaram Behera, Edin Husi\'c, Shweta Jain, Tim Roughgarden, and C.~Seshadhri.
\newblock {FPT} {A}lgorithms for {F}inding {N}ear-{C}liques in {$c$}-{C}losed
  {G}raphs.
\newblock In {\em Proceedings of the 13th Innovations in Theoretical Computer
  Science Conference {(ITCS}~'22)}, volume 215 of {\em LIPIcs}, pages
  17:1--17:24. Schloss Dagstuhl - Leibniz-Zentrum f{\"{u}}r Informatik, 2022.

\bibitem{BB88}
A.~Blokhuis and A.E. Brouwer.
\newblock Geodetic graphs of diameter two.
\newblock {\em Geometriae Dedicata}, 25:527--533, 1988.

\bibitem{BDFH09}
Hans~L. Bodlaender, Rodney~G. Downey, Michael~R. Fellows, and Danny Hermelin.
\newblock On problems without polynomial kernels.
\newblock {\em Journal of Computer and System Sciences}, 75(8):423--434, 2009.

\bibitem{BJK14}
Hans~L. Bodlaender, Bart M.~P. Jansen, and Stefan Kratsch.
\newblock Kernelization lower bounds by cross-composition.
\newblock {\em {SIAM} Journal on Discrete Mathematics}, 28(1):277--305, 2014.

\bibitem{CP17}
Eglantine Camby and Fr{\"{a}}nk Plein.
\newblock A note on an induced subgraph characterization of domination perfect
  graphs.
\newblock {\em Discrete Applied Mathematics}, 217:711--717, 2017.

\bibitem{CHLS13}
Maw{-}Shang Chang, Ling{-}Ju Hung, Chih{-}Ren Lin, and Ping{-}Chen Su.
\newblock Finding large $k$-clubs in undirected graphs.
\newblock {\em Computing}, 95(9):739--758, 2013.

\bibitem{CN85}
Norishige Chiba and Takao Nishizeki.
\newblock Arboricity and subgraph listing algorithms.
\newblock {\em {SIAM} Journal on Computing}, 14(1):210--223, 1985.

\bibitem{CFMPT17}
Alessio Conte, Donatella Firmani, Caterina Mordente, Maurizio Patrignani, and
  Riccardo Torlone.
\newblock Fast enumeration of large $k$-plexes.
\newblock In {\em Proceedings of the 23rd ACM SIGKDD International Conference
  on Knowledge Discovery and Data Mining ({KDD}~'17)}, pages 115--124. {ACM},
  2017.

\bibitem{CMSGMV18}
Alessio Conte, Tiziano~De Matteis, Daniele~De Sensi, Roberto Grossi, Andrea
  Marino, and Luca Versari.
\newblock {D2K:} scalable community detection in massive networks via
  small-diameter $k$-plexes.
\newblock In {\em Proceedings of the 24th {ACM SIGKDD} International Conference
  on Knowledge Discovery and Data Mining ({KDD}~'18)}, pages 1272--1281. {ACM},
  2018.

\bibitem{CK12}
Jean{-}Fran{\c{c}}ois Couturier and Dieter Kratsch.
\newblock Bicolored independent sets and bicliques.
\newblock {\em Information Processing Letters}, 112(8-9):329--334, 2012.

\bibitem{CFK+15}
Marek Cygan, Fedor~V. Fomin, Lukasz Kowalik, Daniel Lokshtanov, D{\'{a}}niel
  Marx, Marcin Pilipczuk, Michal Pilipczuk, and Saket Saurabh.
\newblock {\em Parameterized Algorithms}.
\newblock Springer, 2015.

\bibitem{CGH17}
Marek Cygan, Fabrizio Grandoni, and Danny Hermelin.
\newblock Tight {K}ernel {B}ounds for {P}roblems on {G}raphs with {S}mall
  {D}egeneracy.
\newblock {\em ACM Transactions on Algorithms}, 13(3):43:1--43:22, 2017.

\bibitem{CPPPW11}
Marek Cygan, Geevarghese Philip, Marcin Pilipczuk, Michal Pilipczuk, and
  Jakub~Onufry Wojtaszczyk.
\newblock Dominating set is fixed parameter tractable in claw-free graphs.
\newblock {\em Theoretical Computer Science}, 412(50):6982--7000, 2011.

\bibitem{DF13}
Rodney~G. Downey and Michael~R. Fellows.
\newblock {\em Fundamentals of Parameterized Complexity}.
\newblock Texts in Computer Science. Springer, 2013.

\bibitem{Epp94}
David Eppstein.
\newblock Arboricity and bipartite subgraph listing algorithms.
\newblock {\em Information Processing Letters}, 51(4):207--211, 1994.

\bibitem{ELS13}
David Eppstein, Maarten L{\"{o}}ffler, and Darren Strash.
\newblock Listing all maximal cliques in large sparse real-world graphs.
\newblock {\em {ACM} Journal of Experimental Algorithmics}, 18, 2013.

\bibitem{ES12}
David Eppstein and Emma~S. Spiro.
\newblock The h-{I}ndex of a {G}raph and its {A}pplication to {D}ynamic
  {S}ubgraph {S}tatistics.
\newblock {\em Journal of Graph Algorithms and Applications}, 16(2):543--567,
  2012.

\bibitem{FLZW18}
Qilong Feng, Shaohua Li, Zeyang Zhou, and Jianxin Wang.
\newblock Parameterized algorithms for edge biclique and related problems.
\newblock {\em Theoretical Computer Science}, 734:105--118, 2018.

\bibitem{FRSWW20}
Jacob Fox, Tim Roughgarden, C.~Seshadhri, Fan Wei, and Nicole Wein.
\newblock Finding {C}liques in {S}ocial {N}etworks: {A} {N}ew
  {D}istribution-{F}ree {M}odel.
\newblock {\em {SIAM} Journal on Computing}, 49(2):448--464, 2020.

\bibitem{Gj79}
M.~R. Garey and David~S. Johnson.
\newblock {\em Computers and Intractability: {A} Guide to the Theory of
  NP-Completeness}.
\newblock W. H. Freeman, 1979.

\bibitem{GJS76}
M.~R. Garey, David~S. Johnson, and Larry~J. Stockmeyer.
\newblock Some simplified np-complete graph problems.
\newblock {\em Theoretical Computer Science}, 1(3):237--267, 1976.

\bibitem{GKL12}
Serge Gaspers, Dieter Kratsch, and Mathieu Liedloff.
\newblock On independent sets and bicliques in graphs.
\newblock {\em Algorithmica}, 62(3-4):637--658, 2012.

\bibitem{GV08}
Petr~A. Golovach and Yngve Villanger.
\newblock Parameterized complexity for domination problems on degenerate
  graphs.
\newblock In {\em Proceedings of the 34th International Workshop
  Graph-Theoretic Concepts in Computer Science ({WG}~'08)}, volume 5344 of {\em
  Lecture Notes in Computer Science}, pages 195--205, 2008.

\bibitem{GKS17}
Martin Grohe, Stephan Kreutzer, and Sebastian Siebertz.
\newblock Deciding first-order properties of nowhere dense graphs.
\newblock {\em Journal of the {ACM}}, 64(3):17:1--17:32, 2017.

\bibitem{HKNS15}
Sepp Hartung, Christian Komusiewicz, Andr{\'{e}} Nichterlein, and Ondrej
  Such{\'{y}}.
\newblock On structural parameterizations for the 2-club problem.
\newblock {\em Discrete Applied Mathematics}, 185:79--92, 2015.

\bibitem{HM20}
Danny Hermelin and George Manoussakis.
\newblock Efficient enumeration of maximal induced bicliques.
\newblock {\em Discrete Applied Mathematics}, 303:253--261, 2021.

\bibitem{KMR+22}
Lawqueen Kanesh, Jayakrishnan Madathil, Sanjukta Roy, Abhishek Sahu, and Saket
  Saurabh.
\newblock Further {E}xploiting c-{C}losure for {FPT} {A}lgorithms and {K}ernels
  for {D}omination {P}roblems.
\newblock In {\em Proceedings of the 39th International Symposium on
  Theoretical Aspects of Computer Science ({STACS}~'22)}, volume 219 of {\em
  LIPIcs}, pages 39:1--39:20. Schloss Dagstuhl - Leibniz-Zentrum f{\"{u}}r
  Informatik, 2022.

\bibitem{KR02}
Subhash Khot and Venkatesh Raman.
\newblock Parameterized complexity of finding subgraphs with hereditary
  properties.
\newblock {\em Theoretical Computer Science}, 289(2):997--1008, 2002.

\bibitem{KKNS22}
Tomohiro Koana, Christian Komusiewicz, Andr{\'{e}} Nichterlein, and Frank
  Sommer.
\newblock Covering {M}any (or {F}ew) {E}dges with $k$ {V}ertices in {S}parse
  {G}raphs.
\newblock In {\em Proceedings of the 39th International Symposium on
  Theoretical Aspects of Computer Science ({STACS}~'22)}, volume 219 of {\em
  LIPIcs}, pages 42:1--42:18. Schloss Dagstuhl - Leibniz-Zentrum f{\"{u}}r
  Informatik, 2022.

\bibitem{KKS20}
Tomohiro Koana, Christian Komusiewicz, and Frank Sommer.
\newblock Exploiting $c$-{C}losure in {K}ernelization {A}lgorithms for {G}raph
  {P}roblems.
\newblock In {\em Proceedings of the 28th Annual European Symposium on
  Algorithms ({ESA}~'20)}, volume 173 of {\em LIPIcs}, pages 65:1--65:17.
  Schloss Dagstuhl - Leibniz-Zentrum f{\"{u}}r Informatik, 2020.

\bibitem{KKS21}
Tomohiro Koana, Christian Komusiewicz, and Frank Sommer.
\newblock Essentially {T}ight {K}ernels {F}or ({W}eakly) {C}losed {G}raphs.
\newblock In {\em Proceedings of the 32nd International Symposium on Algorithms
  and Computation ({ISAAC}~'21)}, volume 212 of {\em LIPIcs}, pages
  35:1--35:15. Schloss Dagstuhl - Leibniz-Zentrum f{\"{u}}r Informatik, 2021.

\bibitem{KN21}
Tomohiro Koana and Andr{\'{e}} Nichterlein.
\newblock Detecting and enumerating small induced subgraphs in $c$-closed
  graphs.
\newblock {\em Discrete Applied Mathematics}, 302:198--207, 2021.

\bibitem{K16}
Christian Komusiewicz.
\newblock Multivariate algorithmics for finding cohesive subnetworks.
\newblock {\em Algorithms}, 9(1):21, 2016.

\bibitem{KHMN09}
Christian Komusiewicz, Falk H{\"{u}}ffner, Hannes Moser, and Rolf Niedermeier.
\newblock Isolation concepts for efficiently enumerating dense subgraphs.
\newblock {\em Theoretical Computer Science}, 410(38-40):3640--3654, 2009.

\bibitem{KS15}
Christian Komusiewicz and Manuel Sorge.
\newblock An {A}lgorithmic {F}ramework for {F}ixed-{C}ardinality {O}ptimization
  in {S}parse {G}raphs {A}pplied to {D}ense {S}ubgraph {P}roblems.
\newblock {\em Discrete Applied Mathematics}, 193:145--161, 2015.

\bibitem{L18}
Bingkai Lin.
\newblock The parameterized complexity of the \emph{k}-biclique problem.
\newblock {\em Journal of the {ACM}}, 65(5):34:1--34:23, 2018.

\bibitem{LS21}
Daniel Lokshtanov and Vaishali Surianarayanan.
\newblock Dominating {S}et in {W}eakly {C}losed {G}raphs is {F}ixed {P}arameter
  {T}ractable.
\newblock In {\em Proceedings of the 41st {IARCS} Annual Conference on
  Foundations of Software Technology and Theoretical Computer Science
  ({FSTTCS}~'21)}, volume 213 of {\em LIPIcs}, pages 29:1--29:17. Schloss
  Dagstuhl - Leibniz-Zentrum f{\"{u}}r Informatik, 2021.

\bibitem{Luks82}
Eugene~M. Luks.
\newblock Isomorphism of graphs of bounded valence can be tested in polynomial
  time.
\newblock {\em Journal of Computer and System Sciences}, 25(1):42--65, 1982.

\bibitem{moonM1965}
John~W Moon and Leo Moser.
\newblock On cliques in graphs.
\newblock {\em Israel journal of Mathematics}, 3(1):23--28, 1965.

\bibitem{P03}
Ren{\'{e}} Peeters.
\newblock The maximum edge biclique problem is {NP}-complete.
\newblock {\em Discrete Applied Mathematics}, 131(3):651--654, 2003.

\bibitem{PRS12}
Geevarghese Philip, Venkatesh Raman, and Somnath Sikdar.
\newblock Polynomial {K}ernels for {D}ominating {S}et in {G}raphs of {B}ounded
  {D}egeneracy and {B}eyond.
\newblock {\em {ACM} Transactions on Algorithms}, 9(1):11:1--11:23, 2012.

\bibitem{RS08}
Venkatesh Raman and Saket Saurabh.
\newblock Short {C}ycles {M}ake {W}-hard {P}roblems {H}ard: {FPT} {A}lgorithms
  for {W}-hard {P}roblems in {G}raphs with no {S}hort {C}ycles.
\newblock {\em Algorithmica}, 52(2):203--225, 2008.

\bibitem{SKMN12}
Alexander Sch{\"{a}}fer, Christian Komusiewicz, Hannes Moser, and Rolf
  Niedermeier.
\newblock Parameterized computational complexity of finding small-diameter
  subgraphs.
\newblock {\em Optimization Letters}, 6(5):883--891, 2012.

\bibitem{See96}
Detlef Seese.
\newblock Linear time computable problems and first-order descriptions.
\newblock {\em Mathematical Structures in Computer Science}, 6(6):505--526,
  1996.

\bibitem{ZZ95}
Igor~E. Zverovich and Vadim~E. Zverovich.
\newblock An induced subgraph characterization of domination perfect graphs.
\newblock {\em Journal of Graph Theory}, 20(3):375--395, 1995.

\end{thebibliography}

\newpage
\appendix

\section{Parameter values in Real-World Instances}

\begin{table}[h]
  \caption{A comparison of the number~$n$ of vertices, number~$m$ of edges, the maximum degree~$\Delta$, the closure~$c$, the degeneracy~$d$ and the weak closure~$\gamma$ in social and biological networks.}
  \label{tab:c-closure}

  \centering
  \begin{tabularx}{.7\textwidth}{Xrrrrrr}
    \toprule
    Instance name & $n$& $m$ & $\Delta$ & $c$ & $d$ & $\gamma$\\
    \midrule
    adjnoun-adjacency & 
112 & 425 & 49 & 14 & 6 & 6\\ 
arenas-jazz &
198 & 2\,742 & 100 & 42 & 29 & 18\\ 
ca-netscience & 
379 & 914 & 34 & 5 & 8 & 3\\ 
bio-celegans &
453 & 2\,025 & 237 & 26 & 10 & 9\\
bio-diseasome &
516 & 1\,188 & 50 & 9 & 10 & 5\\
soc-wiki-Vote &
889 & 2\,914 & 102 & 18 & 9 & 8\\ 
arenas-email &
1\,133 & 5\,451 & 71 & 19 & 11 & 8\\
bio-yeast  &
1\,458 & 1\,948 & 56 & 8 & 5 & 4\\
ca-CSphd  &
1\,882 & 1\,740 & 46 & 3 & 2 & 3\\ 
soc-hamsterster  &
2\,426 & 16\,630 & 273 & 77 & 24 & 19\\
ca-GrQc  &
4\,158 & 13\,422 & 81 & 43 & 43 & 9\\
soc-advogato  &
5\,167 & 39\,432 & 807 & 218 & 25 & 21\\
bio-dmela  &
7\,393 & 25\,569 & 190 & 72 & 11 & 12\\
ca-HepPh  &
11\,204 & 117\,619 & 491 & 90 & 238 & 54\\
ca-AstroPh  &
17\,903 & 196\,972 & 504 & 61 & 56 & 30\\ 
soc-brightkite & 
56\,739 & 212\,945 & 1\,134 & 184 & 52 & 49\\
\bottomrule
\end{tabularx}
\end{table}

\end{document}